%% file: main.tex
\documentclass[a4paper,11pt]{scrartcl}

\usepackage{amssymb,amsthm,amsmath, paralist,bbding,pifont}
\usepackage{enumitem}
\usepackage{multicol}
\usepackage{graphicx}
\usepackage{todonotes}
\usepackage[utf8]{inputenc}
\usepackage[normalem]{ulem}

\usepackage{setspace}  
\usepackage{scrpage2}

\usepackage{url}
\usepackage{hyperref}
\usepackage{float}
\usepackage{caption}
\usepackage{ulem}

\usepackage[ruled]{algorithm2e}

\ohead{\pagemark}
\ihead{\headmark}
\ofoot{}

\usepackage[paper=a4paper,left=25mm,right=25mm,top=30mm,bottom=30mm]{geometry}

\newtheorem{theorem}{Theorem}

\newtheorem{remark}[theorem]{Remark}

\usepackage{dsfont}
  \newcommand{\one}{\mathds{1}}

\newif\ifanonymousforreviewreply\anonymousforreviewreplyfalse

\usepackage{amsmath,amssymb,amsthm,nccmath}
\usepackage{nicefrac}
\usepackage{upgreek}
\usepackage{xcolor}
\usepackage{graphicx}
\usepackage{tikz}
\usepackage{subcaption}
\usepackage{cite}
\usepackage{tabu}
\usepackage{comment}

\input{commands.tex}

\title{\Large Special Cosmological Models Derived from the Semiclassical Einstein Equation on Flat FLRW Space-Times}
\author{
\ifanonymousforreviewreply
anonymous authors
\else
Hanno Gottschalk${}^*$, Nicolai Rothe${}^*$ and Daniel Siemssen${}^*$\\ \\{\small ${}^*$School of Mathematics and Natural Science \& IMACM,}\\ {\small University of Wuppertal, D-42119 Wuppertal, Germany}\\
 \texttt{\small $\{$gottschalk,rothe,siemssen$\}$@uni-wuppertal.de}\\[1ex]
\fi
}

\date{{\small \today}}

\begin{document}

\maketitle

\begin{abstract}
    This article presents numerical  work on a special case of the cosmological semiclassical Einstein equation (SCE). The SCE describes the interaction of relativistic quantum matter by the expected value of the renormalized stress-energy tensor of a quantum field with classical gravity. Here, we consider a free, massless scalar field with general (not necessarily conformal) coupling to curvature. In a cosmological scenario with flat spatial sections for special choices of the initial conditions, we observe a separation of the dynamics of the quantum degrees of freedom from the dynamics of the scale factor, which extends a classical result by Starobinski \cite{Starobinski} to general coupling. For this new equation of fourth order governing the dynamics of the scale factor, we study numerical solutions. Typical solutions show a radiation-like Big Bang for the early universe  and de Sitter-like expansion for the late universe. We discuss a specific solution to the cosmological horizon problem that can be produced by tuning parameters in the given equation. Although the model proposed here only contains massless matter, we give a preliminary comparison of the obtained cosmology with the $\Lambda$CDM standard model of cosmology and investigate parameter ranges in which the new models, to a certain extent, is capable of assimilating standard cosmology.   
\end{abstract}{}
\section{Introduction}

This paper introduces a new set of cosmological equations that emerge as a special case from the semiclassical Einstein equation (SCE). The SCE is proposed as a minimal modification to general relativity that takes quantum matter into account, see e.g.\ \cite{birrell1984quantum,fulling1989aspects,Hack:2015zwa,wald1994quantum}. While the SCE is generally not believed to be a fundamental theory, it is widely studied in situations where the relevant physics takes place on scales that are well separated from Planck scale and in cosmological scenarios \cite{schander2021backreaction}. Many special or approximate cosmological solutions to the SCE have been reported \cite{dappiaggi2008stable,Starobinski,JUAREZAUBRY,juarezaubry2020semiclassical,FredenhagenHack2013,FredenhagenHack2013}.

The mathematical understanding of the SCE has only advanced recently for the case of cosmological applications  \cite{pinamonti2011initial,pinamonti2015global,siemssen-gottschalk,meda2020existence}. The intrinsic reason for the difficulties in formulating the SCE in a mathematically consistent fashion lie in the higher derivatives that occur due to the covariant renormalization of the stress-energy tensor \cite{Moretti:2001qh,Hollands:2004yh} leading to an implicit definition of the dynamical system for the SCE \cite{eltzner2011dynamical}. In \cite{siemssen-gottschalk}, however, the dynamical degrees of freedom of the quantum field were redefined using certain expansions of the two point functions of the quantum field in homogeneous distributions as renormalization scheme and inserting correction terms in order to guarantee the equivalence with the standard Hadamard point splitting renormalization. 
In this way, one obtains a formulation of the cosmological SCE as an explicit, infinite-dimensional dynamical system. Here the (infinitely many) dynamic degrees of freedom of the quantum field enter via the prefactors of the aforementioned expansion and are organized in a `tower of moments'. In particular, this is possible for arbitrary (not necessarily conformal) coupling of the quantized field to the scalar curvature.

The present paper is based on a crucial observation, namely that (a) the dynamic equation for the tower of moments is linear homogeneous and that (b) there are massless physical states for all (not necessarily conformal) couplings with vanishing moments as initial conditions. As a homogeneous linear equation maps zero initial data to the vanishing solution, the dynamics of the quantum field effectively decouples from the dynamics of the scale factor. Therefore we derive a fourth order system of equations for the scale factor where quantum effects enter only via geometric terms. This extends the classical Starobinski cosmologies for the conformal coupling \cite{Starobinski} to general couplings. 

The present paper is devoted to the description and the detailed numerical investigation of this new system of cosmological equations. In particular, we review the derivation and the special assumptions on which these equations are based on. Moreover, we give a few examples of explicit solutions for certain settings of the parameters including Minkowski and de Sitter-phases and discuss approximate solutions that relate to higher order gravity \cite{barrow1988inflation,flanagan2003higher}. These latter approximate solutions incorporate both a radiation like Big Bang {\ifanonymousforreviewreply\color{blue}\fi(sufficiently far off the Planck scale)} with a slow-down in expansion speed and a re-acceleration phase prior to an asymptotic de Sitter phase for the late universe. Thereafter we provide numerical evidence that for large regions of the parameter space the numerical solutions for generic parameter settings reproduce this behavior. 

Furthermore we provide a numerical exploration of solutions in dependence of the parameters. Among these solutions we highlight a subset that relates to a solution of the cosmological horizon problem proposed by N.\ Pinamonti \cite{pinamonti2011initial} and based on the diverging negative conformal time for the Big Bang. 

In addition, we discuss parameter settings that assimilate solutions to our equations to solutions of the $\Lambda$CDM standard model of cosmology. This is done for two reasons: On the one hand, it may be viewed as a case study about how flexible cosmological models from the SCE can be, particularly, to anticipate features of more realistic SCE-cosmologies that involve more suitable forms of matter compared to massless scalar quantum fields. On the other hand, we wish to identify some preliminary ideas on the order of magnitude of parameters that enter the SCE. Here, we apply the $\Delta N_\textup{eff}$-test proposed by T.-P.\ Hack as a test for the SCE at redshift factor $z=3000$, corresponding to the emission of the cosmic microwave background \cite{Hack:2015zwa}. In particular, interesting parameter regions seem to be close (but not restricted) to conformal coupling $\xi=\frac{1}{6}$. These first insights of course require confirmation from models with more realistic compositions of matter.

The qualitative results presented in our paper are in line with prior analytical and numerical work by other authors. While the Minkowski solution to the SCE is obvious, de Sitter phases for conformal coupling and massless fields have been found by A.\ A.\ Starobinski \cite{Starobinski}. Recently, special solutions of de Sitter type have been found by B.\ A.\ Juar\'ez for massless and massive quantum fields for special settings of parameters in the Bunch-Davies vacuum state \cite{JUAREZAUBRY}. Here, we find further such solutions for special parameter sets of our new model's equations that lead to de Sitter type expansion \cite{bd_on_ds}.

Asymptotic de Sitter solutions without introducing a cosmological constant have been observed e.g.\ in \cite{dappiaggi2008stable} for a dynamical system derived from the SCE with a massive quantum field for conformal coupling and approximate KMS-like states, see also the in-depth discussion in \cite{Hack:2015zwa} (and references therein) and the recent study of M.\ H\"ansel \cite{haensel2019}  on the phase diagram of the massless SCE with conformal coupling. All these works, however, are restricted to conformal coupling, whereas, in the present article, we extend this type of results to non-conformal coupling. 

Prior work on the numerics of the SCE has been given by P.\ R.\ Anderson \cite{Anderson1,Anderson2,Anderson3,Anderson4}. In contrast to our work, the quantum fields in the first three articles are only corrections to classical background fields. Here, asymptotically classical solutions close to the Big Bang are found for generic values of the renormalization constants. As the numerics in these works is restricted to the very early universe shortly after the Big Bang, this is in line with our findings of a radiation like Big Bang for massless scalar fields (Anderson also treats massive fields).   In \cite{Anderson4}, where no background fields are assumed, de Sitter solutions are found as well. The conformal \emph{ansatz} to set initial conditions at the Big Bang is again restricted to $\xi=\frac{1}{6}$, which is not the case in our system.

Work on the comparison of SCE-cosmology with the $\Lambda$CDM model can be found in \cite{FredenhagenHack2013,Hack:2010iw,hack2013lambda,Hack:2015zwa}, again mostly for conformal coupling. Here, in particular, we employ similar methods for a parameter screening of the SCE in our special system of cosmological equations with more general coupling.  

Our paper is organized as follows: In Section \ref{sec:Short-derivation-of-sce} we recapitulate the moment formulation of the SCE. Section \ref{sec:Decoupling} derives our special, decoupled cosmological models and proves that they lead to full solutions of the SCE. The subsequent Section \ref{sec:General_discussion_on_the_decoupled_cSCE} discusses special de Sitter solutions, the settings for initial conditions and parameters. Also, the state variable is introduced in order to compare the  matter content generated by the quantum field with the matter content of perfect-fluid Friedmann-type cosmologies. 
Finally, Section \ref{sec:General_discussion_on_the_decoupled_cSCE} contains the first numerical results of this paper. In Section \ref{sec:Numeric-solutions-of-the-cosmological-semiclassical-Einstein-equation} we then provide parameter studies for numerical solutions and also include a short digression into the cosmological horizon problem. Furthermore, w show in Section \ref{sec:ComparisonLCDM} that our solutions can be fitted to the $\Lambda$CDM standard cosmology such that they completely lie in the uncertainty band of the $\Lambda$CDM model. Finally, we identify promising regions for parameters using the $\Delta N_\textup{eff}$-test as suggested in \cite{Hack:2015zwa}. Section \ref{sec:Conclusion} contains our conclusions and some comments on future research.

\section{The moment approach to the cosmological SCE}
\label{sec:Short-derivation-of-sce}

This section introduces our notation and briefly recapitulates the moment approach to the cosmological SCE as introduced in \cite{siemssen-gottschalk}. We consider the semiclassical Einstein equation
\begin{equation}
\label{semiclass-einstein-eq}
	G_\munu=\kappa\Tmunuren
\end{equation}
with the metric's sign convention $(-,+,+,+)$. Here, $G_\munu$ is the Einstein tensor and $\Tmunuren$ is the renormalized stress-energy tensor of a free, scalar and chargeless quantum field. The field dynamics is given by the Klein-Gordon (KG) equation,
\begin{equation}
\label{eq:KG}
    \left[\Box+m^2+\xi R\right]\phi=0,
\end{equation}
where $\Box=-g^{\mu\nu}\nabla_\nu\nabla_\mu$ is the d'Alambertian associated with the Levi-Civita connection $\nabla$ for the Lorentzian metric $g_{\mu\nu}$. $\xi\in\R$ parameterizes the curvature coupling and $m\ge0$ defines the field's mass. The special case $\xi=\frac{1}{6}$ is referred to as conformal coupling.  Let $G_\textup{ret/adv}$ be the retarded and advanced fundamental solutions to the KG-equation, then $\phi$ is quantized such that it fulfills the canonical commutation relations (CCR) $[\phi(x),\phi(y)]=i\left(G_\textup{ret}(x,y)-G_\textup{adv}(x,y)\right)$, see e.g.\ \cite{dimock1980algebras,dimock1982classical,fulling1989aspects,wald1994quantum}. Note that with $\phi$, also $\alpha\phi$ for $\alpha\in \R\setminus\{0\}$ is another legitimate local quantum field. By the CCR in the given shape, we are normalizing the field strength to $\alpha=1$ and we obtain $\kappa=8\pi G_N\alpha^2$, where $G_N$ is Newton's gravitational constant. The field strength remains as a free parameter of the theory, for notational brevity, however, we view $\kappa>0$ as the free parameter of the model.

The expectation value of the renormalized stress-energy tensor $\Tmunuren$ is obtained by subtracting the Hadamard parametrix $H(x,y)$ from the two point function $\omega(\phi(x)\phi(y))$ of the quantum fields, applying a certain partial differential operator to $\omega(\phi(x)\phi(y))-H(x,y)$ and performing the point splitting limit $y\to x$, see \cite{bunch1978quantum,Moretti:2001qh}. If $\omega(\phi(x)\phi(y))-H(x,y)$ is infinitely often differentiable, the state $\omega(\cdot)$ is a so-called Hadamard state \cite{fulling1989aspects,radzikowski1996micro}. The Hadamard parametrix is given by the asymptotic expansion of the singular part of $\omega(\cdot)$ in powers of the Synge world function $\sigma(x,y)$ and 
 \begin{equation}
 \label{hadamard-parametrix-equation}
H(x,y)=\lim_{\varepsilon\to+0}\frac{1}{8\pi^2}\left(\frac{\Delta(x,y)^{\nicefrac{1}{2}}}{\sigma(x,y)+i\varepsilon(t(x)-t(y))}+\log\left(\frac{\sigma(x,y)}{\lambda^2}\right)\sum_{j=0}^n\nu_j(x,y)\,\sigma(x,y)^j\right),
 \end{equation}
where $t$ is a time function, $\Delta(x,y)$ is the van Vleck-Morette determinant and $\sigma(x,y)$, for $x,y$ in a geodesically convex neighborhood, is the Synge world function \cite{fulling1989aspects,wald1994quantum}. The coefficient functions $\nu_j(x,y)$
are obtained recursively by the requirement that the Hadamard parametrix (truncated to order $n$) should fulfill the Klein-Gordon Equation \eqref{eq:KG} (up to powers $\sigma^{n+1}(x,y)$) \cite{fulling1989aspects,Moretti:2001qh}.

In the following, we restrict  to flat cosmological space-times $I\times\R^3$, where $I\subseteq \R$ is a time interval and $\R^3$ is the Cauchy surface. Further, the metric is given by 
\begin{equation}
    g=-dt^2+a(t)^2d\vec{x}^2.\label{eq:cosmological-metric}
\end{equation}
Here, $t=t(x)$ is the cosmological time and $a(t)$ is the scale factor. We apply the convention that $t=t_0=0$ and $a(t_0)=1$ stand for the present state of the universe. {\ifanonymousforreviewreply\color{blue}\fi Note that, throughout this article, $a$ is assumed to be sufficiently large in order to not come amiss to the Planck scale. In particular, if we speak of a \emph{Big Bang}, which usually refers to a zero of $a$, we refer to the phase in direct proximity to such a zero but remote enough to justify the expectation that the SCE is still a valid approximation to any sort of underlying theory in that regime. However, in our mostly numeric approach a typical break-down magnitude for a solver is $a\approx10^{\textup{-}7}$ to $a\approx10^{\textup{-}9}$, which is several orders of magnitude larger than the Planck regime.}  

An alternative way to parameterize space-time is by a conformal time coordinate. Using the scale factor $a(t)$, it is given by $ \tau(t)=\int_{t_0}^t\frac{dt}{a(t)}$. Since $d\tau(t)=\frac{dt}{a(t)}$, we see that $-dt^2+a(t)^2d\vec{x}^2$ $=a^2(\tau)(-d\tau^2+d\vec{x}^2)$, from which we instantly derive the conformal equivalence of the metric on the cosmological space-time with the flat metric on a suitable section of Minkowski space. Here, we use the (slightly misleading) convention $a(\tau)$ for the scale factor $a(t)$ at conformal time $\tau=\tau(t)$. 

To study the dynamics of the SCE in the cosmological context, we wish to cast \eqref{semiclass-einstein-eq} in an initial value form. As described in our previous work \cite{siemssen-gottschalk}, this can be achieved via the following procedure:

    (i) We consider fixed time fields and momenta $\varphi(\tau,\vec{x})=a(\tau)\phi(\tau,\vec{x})$, $\pi(\tau,\vec{x})=\partial_\tau\varphi(\tau,\vec{x})$ and a quasi free state $\omega$ evaluated on these fields
    \begin{equation}
    \label{state-dynamics-equation}
        \mathcal{G}(\tau,r)=\left(\begin{array}{c}\mathcal{G}_{\varphi\varphi}(\tau,r)\\
        \mathcal{G}_{(\varphi\pi)}(\tau,r)\\\mathcal{G}_{\pi\pi}(\tau,r)\end{array}\right)=\lim_{\tau'\to\tau}\left(\begin{array}{c}\omega(\varphi(\tau,\vec{x})\varphi(\tau,\vec{y}))\\
        \frac{1}{2}\omega(\varphi(\tau,\vec{x})\pi(\tau,\vec{y})+\pi(\tau,\vec{x})\varphi(\tau,\vec{y}))\\\omega(\pi(\tau,\vec{x})\pi(\tau,\vec{y}))\end{array}\right)
    \end{equation}
    Here, it is assumed that the state $\omega(\cdot)$ is homogeneous on fixed time fields and isotropic on the flat time sections, i.e.\ does only depend on $r=|\vec{x}-\vec{y}|$. Also note that only the symmetric part of the two point function enters \eqref{state-dynamics-equation} as the anti symmetric part is fixed by the CCR. 
    
    (ii) We rewrite the dynamics of the field \eqref{eq:KG} in conformal time as a dynamical equation for $\mathcal{G}(\tau,r)$ and obtain
    \begin{equation}
        \label{eq:dynTwopoint}
        \partial_\tau \mathcal{G}(\tau,r)=\left(\begin{array}{ccc}0&2&0\\\Delta_r-V&0&1\\0&2(\Delta_r-V)&0\end{array}\right)\mathcal{G}(\tau,r)
    \end{equation}
    with $V=(6\xi-1)\frac{a''}{a}+a^2m^2$ and $\Delta_r=r^{-2}\partial_rr^2\partial_r$.
    
    (iii) Considering the corresponding fixed time formulation for $r=|\vec{x}-\vec{y}|$ at conformal time $\tau$
    \begin{align}
    \begin{split}
    \tilde{\mathcal{H}}(\tau,r)&=\left(\begin{array}{c}\tilde{\mathcal{H}}_{\varphi\phi}(\tau,r)\\
        \tilde{\mathcal{H}}_{(\varphi\pi)}(\tau,r)\\
        \tilde{\mathcal{H}}_{\pi\pi}(\tau,r)\end{array}\right)\\
        &=\left(\begin{array}{c}
    a(\tau)^2H_n((\tau,\vec{x}),(\tau,\vec{y})),\\
    \frac{1}{2}\left(\partial_\tau a(\tau)a(\tau')H_n((\tau,\vec{x}),(\tau',\vec{y}))+\partial_{\tau'} a(\tau)a(\tau')H_n((\tau,\vec{x}),(\tau',\vec{y}))\right)_{\tau'=\tau}\\
    \left(\partial_\tau \partial_{\tau'} a(\tau)a(\tau') H_n((\tau,\vec{x}),(\tau',\vec{y}))\right)_{\tau'=\tau}
    \end{array}\right)
    \end{split}
    \end{align}
    of the Hadamard parametrix \eqref{hadamard-parametrix-equation}, we obtain $\langle T_{\mu\nu}^\mathrm{ren}\rangle$, evaluated at conformal time, by applying a partial differential operator $\mathcal{T}_{\mu\nu}$  to $\mathcal{G}-\mathcal{\tilde H}$. After restricting to the diagonal, this yields a $\R^{4\times4}$ valued tensor function. In the following we denote this restriction to $\vec{x}=\vec{y}$ or $r=0$ by $[\,\cdot\,]$. In addition, terms that express renormalization freedom occur. Given that off-diagonal terms of the stress-energy tensor vanish for flat cosmological space-time, we can express the renormalized stress-energy tensor via its energy component $\langle T_{00}^\mathrm{ren}\rangle$ and trace $\langle T^\mathrm{ren}\rangle =g^{\mu\nu}\langle T^\mathrm{ren}_{\mu\nu}\rangle$, see \cite{siemssen-gottschalk}. Therewith,
    \begin{equation}\label{eq:trace-FLRW}\begin{split}
  {\langle T^\mathrm{ren} \rangle} &= \bigl( (6\xi-1) (\xi R + m^2) - m^2 \bigr)\frac{1}{a^2} [\mathcal{G}_{\varphi\varphi}-\tilde{ \mathcal{H}}_{\varphi\varphi}] - \frac{6\xi-1}{a^4} \bigl( [\Delta_r (\mathcal{G}_{\varphi\varphi}-\tilde{ \mathcal{H}}_{\varphi\varphi}) ]
  \\&\quad+ \frac{1}{a^2}[\mathcal{G}_{\pi\pi}-\tilde{ \mathcal{H}}_{\pi\pi}] +\frac{{a'}^2}{a^4}[\mathcal{G}_{\varphi\varphi}-\tilde{\mathcal{H}}_{\varphi\varphi}]-2\frac{a'}{a^3}[\mathcal{G}_{(\varphi\pi)}-\tilde{\mathcal{H}}_{(\varphi\pi)}]\bigr)\\
  &\quad- \frac{9\xi-2}{2\uppi^2} [v_1] + 4 c_1 m^4 - c_2 m^2 R - (6 c_3 + 2 c_4) \Box R,
\end{split}\end{equation}
where $R = 6 \frac{a''}{a^3} $, $\Box R = 36 \frac{a'' a^{\prime\,2}}{a^7} - 18 \frac{a^{\prime\prime\,2}}{a^6} - 24 \frac{a^{(3)} a'}{a^6} + 6 \frac{a^{(4)}}{a^5} $ and
$[v_1]$ is the conformal anomaly
\begin{equation}\begin{split}
  [v_1] &= \frac{m^4}{8} + \frac{1}{60} \left( \frac{a^{\prime\,4}}{a^8} - \frac{a'' a^{\prime\,2}}{a^7} \right) + \frac{(6\xi-1) m^2}{4} \frac{a''}{a^3} + \frac{(6\xi-1)^2}{8} \frac{a^{\prime\prime\,2}}{a^6} \\&\quad + \frac{5\xi-1}{20} \left( 6 \frac{a'' a^{\prime\,2}}{a^7} - 3 \frac{a^{\prime\prime\,2}}{a^6} - 4 \frac{a^{(3)} a'}{a^6} + \frac{a^{(4)}}{a^5} \right).
\end{split}\end{equation}
Moreover, 
\begin{equation}\label{eq:T00-FLRW}\begin{split}
  {\langle T_{00}^\mathrm{ren} \rangle} &= \frac{1}{2} [ \mathcal{G}_{\pi\pi}-\tilde{ \mathcal{H}}_{\pi\pi}] - \frac{1}{2a^2} [\Delta_r (\mathcal{G}_{\varphi\varphi}-\tilde{ \mathcal{H}}_{\varphi\varphi})] + \frac12  m^2 [\mathcal{G}_{\varphi\varphi}-\tilde{ \mathcal{H}}_{\varphi\varphi}] \\&\quad + \xi \left( \frac{G_{00}}{a^2} [\mathcal{G}_{\varphi\varphi}-\tilde{ \mathcal{H}}_{\varphi\varphi}] + 6 \frac{a'}{a} [\mathcal{G}_{(\varphi\pi)}-\tilde{ \mathcal{H}}_{(\varphi\pi)}]-6\frac{{a'}^2}{a^2}[\mathcal{G}_{\varphi\varphi}-\tilde{\mathcal{H}}_{\varphi\varphi}] \right) \\&\quad - \frac{a^2}{4\uppi^2} [v_1] - c_1 a^2 m^4 + c_2 m^2 G_{00} + (3c_3 + c_4) J_{00}.
\end{split}\end{equation}
with  $G_{00} = 3 \frac{a^{\prime\,2}}{a^2}$ and $J_{00} = -24 \frac{a'' a^{\prime\,2}}{a^5} - 6 \frac{a^{\prime\prime\,2}}{a^4} + 12 \frac{a^{(3)} a'}{a^4}$.
    
    (iv) One of the problems with the Hadamard parametrix $\tilde{\mathcal{H}}(\tau,r)$ is that it does not fulfill a well-defined set of dynamic equations. Therefore we introduce an auxiliary (non-covariant) parametrix
    \begin{equation}
        \mathcal{H}_n(\tau,r)=\left(\begin{array}{c}0\\0\\\gamma_{-1}(\tau)\end{array}\right) h_{-2}(r)+\sum_{l=0}^n\left(\begin{array}{c}\alpha_j(\tau)\\\beta_j(\tau)\\\gamma_{j}(\tau)\end{array}\right) h_{2j}(r)
    \end{equation}   
    with the homogeneous distributions $h_z(r) = \frac{e^{iz \pi/2}}{2\pi^2} \frac{r^{z-2}}{\Gamma(z)} \left( \log\Bigl(\frac{r}{\mu}\Bigr) - \psi(z) \right)$ defined for $z\in \mathbb{C}$ via analytic continuation and depending on some parameter $\mu>0$. Here, $\psi(z)$ denotes the Digamma function. Using $\Delta_r h_j(r)=h_{j-2}(r)$, we obtain the coefficient functions $\alpha_l(\tau)$, $\beta_l(\tau)$ and $\gamma_l(\tau)$ recursively by the starting condition $\gamma_{-1}=\frac12$, $\alpha_0=\frac12$ and $\beta_0=0$ and the equation 
    \begin{equation}
    \label{eq:dynModHadamard}
    \partial_\tau \mathcal{H}_n- \left(\begin{array}{ccc}0&2&0\\\Delta_r-V&0&1\\0&2(\Delta_r-V)&0\end{array}\right)\mathcal{H}_n(\tau,r)   =\mathcal{O}(r^{2(n-1)}).
    \end{equation}
    Then, we can rewrite expressions like $[\mathcal{G}_\sharp-\tilde{\mathcal{H}}_{\sharp,n}]$ as $[\mathcal{G}_\sharp-\mathcal{H}_{\sharp, n}]+[\mathcal{H}_{\sharp,n}-\tilde{\mathcal{H}}_{\sharp,n}]$, $\sharp\in \{\varphi\varphi,(\varphi\pi),\pi\pi\}$, or $[\Delta_r(\mathcal{G}_{\sharp}-\tilde{\mathcal{H}}_{\sharp,n})]$ as $[\Delta_r(\mathcal{G}_{\sharp}-\mathcal{H}_{\sharp, n})]+[\Delta_r(\mathcal{H}_{\sharp,n}-\tilde{\mathcal{H}}_{\sharp,n})]$. The second term in these sums can be evaluated explicitly in terms of the scale function $a(\tau)$ and its derivatives up to order four as long as the order $n$ is chosen larger or equal $2$.
    
    (v) We define a sequence of so-called moment functions 
    \[
    \mathcalm_{n,\sharp}=[\Delta_r^n(\mathcal{G}_{\sharp}-\mathcal{H}_{\sharp, j})],
    \]
    $\sharp\in \{\varphi\varphi,(\varphi\pi),\pi\pi\}$, and arrange these three real-valued functions of time into $\R^3$-valued functions  $\mathcalm_n=(\mathcalm_{n,\varphi\varphi},\mathcalm_{n,(\varphi\pi)},\mathcalm_{n,\pi\pi})^\top\in\R^3$, which are independent of $j$ provided that $j\geq n+1$. From \eqref{eq:dynTwopoint} and \eqref{eq:dynModHadamard} we deduce the following recursive set of equations
    \begin{equation}
    \label{eq:dynMomentRecursive}
    \partial_\tau \mathcalm_n=A\mathcalm_n+B\mathcalm_{n+1} \quad\text{with}\quad A= \left(\begin{array}{ccc}
    0 & 2 & 0 \\
    -V & 0 & 1 \\
    0 & -2V & 0
  \end{array}\right)\quad\text{and}\quad B=\left(\begin{array}{ccc}0 & 0 & 0 \\
    1 & 0 & 0 \\
    0 & 2 & 0\end{array}\right).
    \end{equation}
    Introducing sequences $\mathcalm=(\mathcalm_n)$ in weighted, discrete $L^p$-spaces $ 
  \vec\ell^p(w) = \mathbb{R}^3 \otimes \ell^p(w)$ with weights $w_n=w^{-n}$, $w>1$, we obtain the dynamical system
  \begin{equation}
      \label{eq:dynMoments}
      \partial_\tau\mathcalm =(A\otimes \one+B\otimes \mathbb{L})\mathcalm,
  \end{equation}
  where $\mathbb{L}$ is the left-shift operator on $\vec\ell^p(w)$. It has been shown that this infinite dynamical system has maximal solutions in conformal time $\tau$ for any four-times-differentiable scale function $a(\tau)$.

(vi) We consider the energy and the trace equation derived from the semiclassical Einstein equation \eqref{semiclass-einstein-eq}
\begin{equation}
	 -R=g^\munu G_\munu=\kappa\Tren
	 \hspace{1cm}\textup{and}\hspace{1cm} 
	 G_{00}=\kappa\,\Tmunurenen, 
\end{equation}
respectively. Wrapping up (i)--(v) above, one obtains 
\begin{align}
 \begin{split}\label{eq:SCEMomerntTrace}
  0 &=
  \left( -12 (3 c_3 + c_4) -\frac{1}{480\pi^2} + \frac{6\xi-1}{48\pi^2} + \frac{(6\xi-1)^2}{16\pi^2} \log(a \lambda_0) \right)
  \\
  &\hspace{5.5cm}\cdot\left( \frac{a^{(4)}}{a^5} - 4 \frac{a^{(3)} a'}{a^6} - 3 \frac{(a'')^2}{a^6} + 6 \frac{a'' (a')^2}{a^7} \right) 
  \\
  & + \frac{(6\xi-1)^2}{32\pi^2} \left( 4 \frac{a^{(3)} a'}{a^6} + 3 \frac{(a'')^2}{a^6} - 10 \frac{a'' (a')^2}{a^7} \right) + \frac{1}{240\pi^2} \left( -\frac{a'' (a)^2}{a^7} + \frac{(a')^4}{a^8} \right) 
  \\
  & + \left( \frac{6}{\kappa} + m^2 \Bigl( -6 c_2 + \frac{1}{48\pi^2} + \frac{6\xi-1}{8\pi^2} \bigl(1 + \log(a \lambda_0) \bigr) \Bigr) \right) \frac{a''}{a^3} 
  \\
  & + \frac{(6\xi-1) m^2}{16\pi^2} \frac{(a')^2}{a^4} + m^4 \left( 4 c_1 + \frac{1}{32\pi^2} + \frac{1}{8\pi^2} \log(a \lambda_0) \right) 
  \\
  &  - \frac{m^2}{a^2} \mathcalm_{\varphi\varphi,0}+ (6\xi-1) \left( \Big( 6\xi \frac{a''}{a^5} - \frac{(a')^2}{a^6} + \frac{m^2}{a^2} \Big) \mathcalm_{\varphi\varphi,0}\right. 
  \\
  &\hspace{5.5cm}\left.+ 2 \frac{a'}{a^5} \mathcalm_{(\varphi\uppi),0} - \frac{1}{a^4} \big( \mathcalm_{\uppi\uppi,0} + \mathcalm_{\varphi\varphi,1} \big) \right)
 \end{split}
\end{align}
for the trace equation and 

 \begin{align}\label{eq:SCEMomerntEnergy}
 \begin{split}
  0 &= \left( 6 (3 c_3 + c_4) + \frac{1}{960\pi^2} - \frac{6\xi-1}{96\pi^2} - \frac{(6\xi-1)^2}{32\pi^2} \log(a \lambda_0) \right)
  \\
  &\hspace{6.5cm}\cdot \left( 2 \frac{a^{(3)} a'}{a^4} - \frac{(a'')^2}{a^4} - 4 \frac{a'' (a')^2}{a^5} \right) 
  \\
  & - \frac{(6\xi-1)^2}{16\pi^2} \frac{a''(a')^2}{a^5} + \frac{1}{960\pi^2} \frac{(a')^4}{a^6} - m^4 \left( c_1 + \frac{1}{32\pi^2} \log(a \lambda_0) \right) a^2 
  \\
  &+ \left( -\frac{3}{\kappa} + m^2 \Bigl( 3 c_2 - \frac{1}{96\pi^2} - \frac{6\xi-1}{16\pi^2} \big(1 + \log(a \lambda_0) \big) \Big) \right) \frac{(a')^2}{a^2} 
  \\
  & + \frac{m^2}{2} \mathcalm_{\varphi\varphi,0} + (6\xi-1) \left( -\frac{(a')^2}{2a^4} \mathcalm_{\varphi\varphi,0} + \frac{a'}{a^3} \mathcalm_{(\varphi\uppi),0} \right)
  \\
  & \hspace{7.5cm} + \frac{1}{2a^2} \big( \mathcalm_{\uppi\uppi,0} - \mathcalm_{\varphi\varphi,1} \big)
  \end{split}
\end{align}
for the energy constraint, see \cite{siemssen-gottschalk} for the details of the calculation. We note that the respective first lines of \eqref{eq:SCEMomerntTrace} and \eqref{eq:SCEMomerntEnergy} only consist of quantum contributions of the field, that is, of terms originating in the renormalization freedom and the trace anomaly as well as expicitly state-dependent contributions (the log-terms).

While the (infinite dimensional) dynamical system from \eqref{eq:dynMoments} and \eqref{eq:SCEMomerntTrace} is well posed for any set of initial conditions $(a,a^{(1)},a^{(2)}, a^{(3)})^\top\in \R^3$ and $\mathcalm\in\vec{\ell} ^p(w)$, it is, however, not clear whether there exists a Hadamard state $\omega$ for a given set of moments $\mathcalm$. Let us therefore shortly comment on physical initial conditions from the `tow-in' technique as described in \cite{siemssen-gottschalk} that guarantees the existence of physical solutions for at least a subset of moments. For this purpose, a Hadamard state and the corresponding tower of moments are prepared on some simple space-time, e.g.\ Minkowski space-time. Then, after a short waiting time, the space-time is deformed by an auxiliary dynamical equation that `tows' the vector of initial conditions $(a(\tau),\ldots,a^{(3)}(\tau))^\top$ to some desired vector of initial conditions $(a_0,\ldots,a_3)^\top\in\R^4$. Both the Hadamard state and the tower of moments propagate forward accordingly. After the tow-in phase, an interpolation phase follows where the auxiliary dynamics of $(a(\tau),\ldots,a^{(3)}(\tau))^\top$ and $\mathcalm(\tau)$ is quickly interpolated to the dynamics of the SCE given by \eqref{eq:dynModHadamard} and \eqref{eq:SCEMomerntTrace}. Thereafter the system follows this dynamic. It can be shown that the latter can be done in a way that (a) the energy constraint \eqref{eq:SCEMomerntEnergy} and thereby the full SCE is fulfilled and (b) the initial conditions with respect to the dynamics $a(t)$ lie in an $\varepsilon$-neighbourhood to $(a_0,\ldots,a_{(3)})^\top$ for arbitrarily small $\varepsilon>0$. For the details, we refer to \cite[Thm. 5.11]{siemssen-gottschalk}.  
 
 As the last statement of this preparatory section, we present the tower of moments for the Minkowski state with scale factor $a(\tau)=1$. As computed in \cite[(4.8)]{siemssen-gottschalk}, the moments $\mathcalm$ in this case are given by
 \begin{align}\label{eq:Moments-vacuum}
 \begin{split}
  \mathcalm_{\varphi\varphi,n} &= \frac{1}{2\uppi^2} \bigl(\tfrac12 m\bigr)^{2n+2} \Bigl( \log(\tfrac12 m \mu) + \psi(2n+2) - \tfrac12 \bigl( \psi(n+1) + \psi(n+2) \bigr)\Bigr) \binom{2n+1}{n+1}, \\
  \mathcalm_{(\varphi\pi),n} &= 0, \\
  \mathcalm_{\pi\pi,n} &= \frac{1}{\uppi^2} \bigl(\tfrac12 m\bigr)^{2n+4} \Bigl( \log(\tfrac12 m \mu) + \psi(2n+2) - \tfrac12 \bigl( \psi(n+1) + \psi(n+3) \bigr)\Bigr) \frac{(2n+1)!}{n!(n+2)!}.
\end{split}
\end{align}
It has been shown that $\mathcalm\in\vec{\ell}^p(w)$ for sufficiently large weights $w$.




\section{The cSCE with zero mass and Minkowski-vacuum-like states as a dynamical system \label{sec:Decoupling}}

Two observations in the dynamics of the moments in \eqref{eq:dynMoments} and the formula for the moments of the Minkowski vacuum state in \eqref{eq:Moments-vacuum} are remarkable: At first, \eqref{eq:dynMoments} is a linear homogeneous differential equation. At second, the moments for the Minkowski vacuum state vanish for $m=0$, i.e.\ $\mathcalm=0$ is fulfilled at the initial point of the `tow-in' process, and hence $\mathcalm(\tau)=0$ holds on the entire cosmological space-time with expansion history $a(\tau)$ that partially consists of the tow-in phase and partially of the SCE phase. Thus, the quantum state completely decouples from the dynamics of the space-time. In this case, all terms in \eqref{eq:SCEMomerntTrace} and \eqref{eq:SCEMomerntEnergy} that are proportional to $m^2$, $m^4$ and $\mathcalm$ are eliminated which largely simplifies our equations. Additional justification that this procedure actually results in physical solutions is given in Theorem 1 below. Finally, one obtains a fourth-order ODE for the scale factor $a(\tau)$ together with a third-order constraint.

Furthermore, since we are interested in cosmology including solutions with a Big Bang, we re-express the dynamic equations for $a(\tau)$ given in conformal time $\tau$ in cosmological time $t$. This is done in order to deal with Big Bang-solutions, as in some cases a Big Bang-event $a(t)=0$ is shifted to conformal time $\tau=-\infty$, see also Subsection \ref{sec:Horizon}. 

Formally, we substitute $\frac{d}{d\tau}=a(t) \frac{d}{d t}$ into the trace equation\footnote{Here we employ the convention that an expression $a^{(k)}$ in an equation with dot-derivatives denotes the $k$-th derivative w.r.t.\ cosmological time whereas in an equation with prime derivatives the same expression stands tor a $k$-th conformal derivative. The same convention applies to initial conditions $(a_0,\ldots,a_3)$.  } and obtain
\begin{align}
\begin{aligned}
    0 &= \left(k_2\log(\lambda_0\,a)-k_1\right)
                 \left(\frac{a^{(4)}}{a} + 3\frac{\dot{a}a^{(3)}}{a^2} + \frac{\ddot{a}^2}{a^2} - 5\frac{\dot{a}^2\ddot{a}}{a^3}\right)\\
       &\hspace{1cm}    + \frac{k_2}{2}\left(4\frac{\dot{a}a^{(3)}}{a^2} + 3\frac{\ddot{a}^2}{a^2} + 12\frac{\dot{a}^2\ddot{a}}{a^3} - 3\frac{\dot{a}^4}{a^4}\right)
       - k_3\frac{\dot{a}^2\ddot{a}}{a^3}+k_4\left(\frac{\dot{a}^2}{a^2}+\frac{\ddot{a}}{a}\right),
\end{aligned}\label{eq:dyn-sys-trace-eq-decoupled}
\end{align}
where we use the dot as a symbol of derivatives w.r.t.\ cosmological time and for the ease of notation we introduced
\begin{align}\label{eq:k-constants}
\begin{split}
	k_1&=12(3c_3+c4) + \frac{1}{480\pi^2} - \frac{6\xi - 1}{48\pi^2},\\
    k_2&=\frac{(6\xi-1)^2}{16\pi^2}\ge0,\hspace{1cm}
	k_3=\frac{1}{240\pi^2}>0,\hspace{1cm}
	k_4=\frac{6}{\kappa}>0.
\end{split}
\end{align}
The parameters $k_1,k_2,k_3$ are dimensionless, but a numerical value of $k_4$ depends on the chosen unit system. Moreover, as noted before, $k_1$ consists only of quantum contributions. The energy constraint in the present setting reads
\begin{equation}
\begin{aligned}
    0 &= -\left(k_2\log(\lambda_0\,a)-k_1\right)
                 \left(\dot{a}a^{(3)}-\frac{1}{2}\ddot{a}^2+\frac{\dot{a}^2\ddot{a}}{a}-\frac{3}{2}\frac{\dot{a}^4}{a^2}\right)\\
                 &\hspace{1cm}-k_2\left(\frac{\dot{a}^2\ddot{a}}{a}+\frac{\dot{a}^4}{a^2}\right)+ \frac{k_3}{4}\frac{\dot{a}^4}{a^2}-\frac{k_4}{2}\dot{a}^2.
\end{aligned}\label{eq:dyn-sys-en-eq-decoupled}
\end{equation}

Let us next reconsider the `tow-in' procedure for the proof of the existence of a physical Hadamard state corresponding to some dynamics of moments $\mathcalm(\tau)$, for the special case that we start the tow-in process with $\mathcalm=0$ and hence obtain $\mathcalm(\tau)=0$. For that scenario, we can refine the results of \cite{siemssen-gottschalk} on the existence of physically meaningful solutions to the cSCE in an arbitrarily small neighborhood of the initial conditions  $(a_0,\ldots,a_3)^\top$ for $a(t)$ at $t=0$ and its first to third derivatives. In the present context, we modify the tow-in argument and prove that any set of initial conditions $(a_0,\ldots,a_3)^\top$ with $k_2\log(\lambda_0 a_0)-k_1\neq0$ can be matched \emph{exactly}:

\begin{theorem}
\label{thm:initCond}
Let $(a_0,\ldots,a_3)^\top\in (0,\infty)\times \R^3$ be initial values for $a(t)$ at cosmological time $t=0$ such that $(k_2\log(\lambda_0 a_0)-k_1)\not=0$. Then the following holds:
\begin{itemize}
    \item[\textup{(i)}] There exists $a(t)$, a unique solution to the ODE \eqref{eq:dyn-sys-trace-eq-decoupled} on some interval of time $(t_{\textup{i}},t_{\textup{f}})$, $t_{\textup{i}}\in[-\infty,0)$, $t_{\textup{f}}\in(0,\infty]$ such that $a(t)$ fulfills \eqref{eq:dyn-sys-trace-eq-decoupled} with the given initial conditions.
    \item[\textup{(ii)}] If the initial conditions fulfill the energy constraint \eqref{eq:dyn-sys-en-eq-decoupled} at $t=0$, then \eqref{eq:dyn-sys-en-eq-decoupled} is fulfilled for all times.
    \item[\textup{(iii)}] There exists a Hadamard state on the cosmological space-time defined by $a(t)$, $t\in(t_{\textup{i}},t_{\textup{f}})$ for the massless Klein-Gordon field with associated tower of moments fulfilling $\mathcalm(t)=0$.
\end{itemize}
 Hence, any cosmological space-time defined by $a(t)$ for $t\in(t_{\textup{i}},t_{\textup{f}})$  as described in \textup{(i)} and \textup{(ii)} is a solution to the cSCE for a Hadamard state as in \textup{(iii)}.
\end{theorem}
\begin{proof}
Note that by the assumption $(k_2\log(\lambda_0 a_0)-k_1)\not=0$ equation \eqref{eq:dyn-sys-trace-eq-decoupled} can be brought to the form 
$$
a^{(4)}=f(a,\dot a,\ddot a, a^{(3)})
$$
where $f(\cdot)$ is locally Lipschitz except for $a=0$ and $(k_2\log(\lambda_0 a(t))-k_1)=0$. Therefore, assertion (i) follows from standard theory of ODE, see e.g.\ \cite{agarwal2008introduction}.

 Since the energy equation is a constant of motion for the trace equation, statement (ii), is well known, see, e.g. eq.\ \eqref{eq:EnergyContraint} below.

To prove (iii), we modify the `tow-in' argument from Section \ref{sec:Short-derivation-of-sce} in the following way: let $(t_{\textup{i}},t_{\textup{f}})$ be an interval containing $t=0$ such that the solution $a(t)$ from (i) is given. 

Consider the switching function $\chi\in C^\infty(\R,[0,1])$ with the property $\chi(t)=0$ for $t<\frac{3}{4}t_{\textup{i}}$ and $\chi(t)=1$ for $t>\frac{1}{4}t_{\textup{i}}$. Moreover, let the cosmological space-time be defined by the smooth scale factor
$$
a_\text{tow}(t)=\chi(t)a(t)+(1-\chi(t)),$$
which is Minkowski for $t\leq\frac{3}{4}t_{\textup{i}}$. 

Thus we can consider $\omega_\text{vac}$, the Minkowski vacuum state for the massless free field for values $t<\frac{3}{4}t_{\textup{f}}$, which is propagated forward to a state $\omega_\text{tow}$ on the entire (globally hyperbolic) space-time defined by $a_\text{tow}$ via the massless Klein-Gordon dynamics. As the Minkowski vacuum state $\omega_\text{vac}$ is Hadamard and $a_\text{tow}(t)$ is smooth, so is the propagated state $\omega_\text{tow}$ \cite{siemssen-gottschalk}. 

Furthermore, the tower of moments $\mathcalm_\text{vac}(t)$ associated to $\omega_\text{vac}$ for $t<\frac{3}{4}t_{\textup{i}}$ fulfills $\mathcalm_\text{vac}(t)=0$ (cf.\ \eqref{eq:Moments-vacuum} with $m=0$) and therefore, by \eqref{eq:dynMoments}, the tower of moments associated with $\omega_\text{tow}$ satisfies $\mathcalm_\text{tow}(t)=0$, also for $t\in(t_{\textup{i}},t_{\textup{f}})$. By this circumstance, the cSCE holds on $(\frac{1}{4}t_{\textup{i}},t_{\textup{f}})$. Lastly, if the state is defined on this interval of time, it can be propagated backwards to a state $\omega$ on the (also globally hyperbolic) space-time defined by $a(t)$, $t\in(t_{\textup{i}},t_{\textup{f}})$ which, for the same reasons as above, results is a Hadamard state on this cosmological space-time. Here again the associated tower of moments fulfills $\mathcalm(t)=0$, for $t\in (t_{\textup{i}},t_{\textup{f}})$ as $\mathcalm(t)=\mathcalm_\text{tow}(t)=0$ for $t>\frac{1}{4}t_{\textup{i}}$ and \eqref{eq:dynMoments}. Thus $\omega$, $a(t)$ and $\mathcalm(t)$ satisfy the cSCE.  This proves the third assertion.       
\end{proof}

Let us shortly compare Theorem \ref{thm:initCond} to the well-known decoupling result for massless conformal fields by Starobinski \cite{Starobinski}. In the case of conformal coupling, Starobinski's result is more general, as \emph{every} state decouples from the cSCE, up to the conformal anomaly term. Our result is restricted to a special class of towed-in massless Minkowski vacuum states, exclusively. On the other hand, our result is more general as conformal coupling $\xi=\frac{1}{6}$ is not required.

A last remark in the present section concerns the role of the regularization parameter $\lambda_0$. Since it is only used to construct the auxiliary parametrices, it does not bear any physical meaning. Nevertheless, different values for $\lambda_0$ lead to different solutions $a(t)$. However, note that $\omega_\text{vac}$ has to be towed in via the $\lambda_0$-dependent space-time defined by $a_\text{tow}(t)$ and therefore also the state $\omega_\text{tow}$ on $(\frac{3}{4}t_{\textup{i}},t_{\textup{f}})$ implicitly depends on $\lambda_0$. 

\section{General discussion on the decoupled cSCE}
\label{sec:General_discussion_on_the_decoupled_cSCE}
One can easily see that the trace equation \eqref{eq:dyn-sys-trace-eq-decoupled}, with the energy equation \eqref{eq:dyn-sys-en-eq-decoupled} regarded as an algebraic constraint (particularly, on the initial values), and the energy equation, regarded as an ODE in its own right, are equivalent under the assumption $\dot{a}\neq 0$. This observation can be traced back to a property of the Einstein tensor's components in FLRW space-time, namely that
\begin{equation}
\label{eq:EnergyContraint}
	\ddt\big(a^2\, G_{00}\big)+2a\dot{a}G_{00}=-a\dot{a}~g^\munu G_\munu,    
\end{equation}
and thus, imposed by the SCE \eqref{semiclass-einstein-eq}, the analog equation holds for $\Tmunuren$ as well. However, due to the latter restriction $\dot{a}=0$ and in order to avoid numerical difficulties close to the $\dot{a}=0$\,-regime, we prefer to work with the trace equation \eqref{eq:dyn-sys-trace-eq-decoupled}. By the aforementioned equivalence, we then conclude that choosing suitable initial conditions to fulfill \eqref{eq:dyn-sys-en-eq-decoupled} results in solutions which satisfy \eqref{eq:dyn-sys-en-eq-decoupled} for all times.

Some exact solutions can be found by the ansatz $a(t)=\exp(H t)$. Inserting it into either the trace equation or the energy constraint, we obtain a fourth order polynomial equation for $H$ solved by 
\begin{equation}
\label{eq:deSitterSolutions}
	H=0\hspace{1cm}\textup{or by}\hspace{1cm}H=\pm H^{\textup{dS}}:=\pm\sqrt{\tfrac{2k_4}{k_3-8k_2}}.    
\end{equation}
Obviously, $H=0$ stands for the Minkowski solution while $\pm H^\textup{dS}$ are expanding/shrinking de Sitter solutions with constant Hubble parameter $H=\frac{\dot{a}(t)}{a(t) }$. Note that $H^\textup{dS}$ is a real number if and only if $k_2<\frac{k_3}{8}$, or equivalently, $|\xi-\frac{1}{6}|<\sqrt{\nicefrac{1}{4320}}$. The symmetric occurrence of expanding and shrinking de Sitter solutions is a consequence of the time reflection invariance $t\to -t$ of the cSCE which can be easily read off from the decoupled equations and which will be furtherly exploited below.

\begin{figure}
    \centering
    \begin{minipage}{.4\textwidth}
    \includegraphics{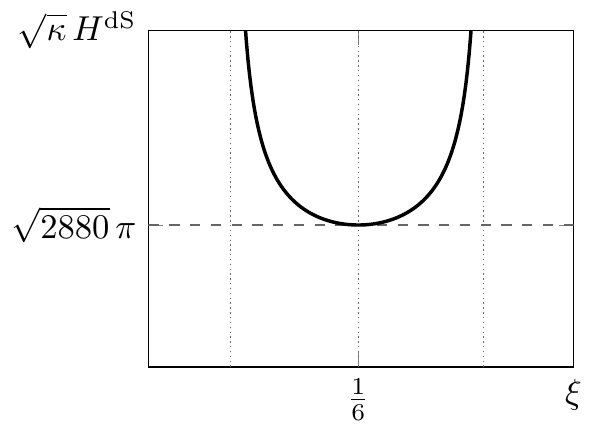}
    \end{minipage}\quad
    \begin{minipage}{.55\textwidth}
    \captionsetup{format=plain, labelfont=bf}
    \caption{\hspace{2pt}The (constant) Hubble rate $H^\textup{dS}$ of the de Sitter solutions defined in \eqref{eq:deSitterSolutions}, shown as a function of $\xi$. The vertical axis is rescaled by $\sqrt{\kappa}$ as $H^\textup{dS}$ is proportional to this value and the graphic depends on no other parameter. The dotted vertical lines show the distinguished values $\xi\in\{\frac{1}{6},\frac{1}{6}\pm\sqrt{\nicefrac{1}{4320}}\}$. We see the analog graphic to the massless, $\Lambda=0$-case of \cite{bd_on_ds}.\label{fig:HdS_of_xi}} 
    
    ~
    \end{minipage}
    
\end{figure}

We note that the de Sitter solutions found here are not necessarily identical to those discussed in \cite{JUAREZAUBRY}, as the `tow-in' states we consider here are constructed very differently from the Bunch-Davies state for the massless field. A complete list of de Sitter solutions based on Bunch-Davis states (massless and massive) is given in \cite{bd_on_ds}. 

\begin{remark}
\normalfont
Note the similarity of Figure \textup{\ref{fig:HdS_of_xi}} with a particular graphic in \textup{\cite{bd_on_ds}}, namely the one showing the de Sitter Hubble rate $H$ as a function of $\xi$ in the massless $\Lambda=0$-case. Despite the different choice of state in \textup{\cite{bd_on_ds}} a different state was chosen, the polynomial equation to be solved for $H$ is very similar. Particularly, the same analysis as in \textup{\cite{bd_on_ds}} may be performed for our `tow-in'-states, also with $\Lambda\neq0$, and we would analogously observe a quantum branch and a $($semi-$)$ classical branch of de Sitter solutions for $\Lambda>0$.
\end{remark}

The problem further simplifies if $\xi=\frac{1}{6}$ and $3c_3+c_4=-\frac{1}{5760\pi^2}$, or $k_1=k_2=0$, respectively, as e.g.\ considered by Starobinski \cite{Starobinski}. In this setting it suffices to take into account the energy constraint. Then, the latter reads as
\begin{equation}
   0=\frac{k_3}{4}\,\frac{\dot{a}^4}{a^2}-\frac{k_4}{2}\dot{a}^2 
\end{equation}
and is solved by either $\dot{a}(t)=0$ or by $\dot{a}(t)=H\,a(t)$ with $H=\pm\sqrt{\frac{2k_4}{k_3}}=\pm H^\textup{dS}$. 
In this scenario, the Minkowski and de Sitter solutions thus are the only ones.

In the general case, we solve the decoupled cSCE numerically. This requires the specification of initial conditions and insight into the dependency on the parameters $k_1,\ldots,k_4$. Let us start with a discussion of the initial conditions, a parameter study will be done in the subsequent section. 

At first, we note that the set of solutions of both our equations are invariant under transformations of the form 
\begin{equation}
    a(t)~\mapsto~\beta_1\,a(\beta_2\,t+\beta_3)\label{eq:rescaling-invariance}
\end{equation}
($\beta_1,\beta_2,\beta_3\in\R,~\beta_2\neq0$), at least with simultaneous redefinitions 
\[
k_4\mapsto\beta_2^2\,k_4\qquad\textup{and}\qquad k_1\mapsto k_1-k_2\log(\beta_1).
\]
Particularly, for a full study of initial conditions and parameters we can fix our initial time to be zero at the present time and our initial value of the scale factor of the present universe to $a(0)=a_0=1$.

A physical initial value for $\dot{a}(t)$ is the present day Hubble constant $H_0$, which is $2.2 \times 10^{-18}$ sec${}^{-1}$ in SI-units or $1.19\times 10^{-61}$ in Planck units \cite{pdg-data}. However, by the invariance of our equations under \eqref{eq:rescaling-invariance} this value is rather arbitrary and should be viewed as a physically realistic choice.

For the initial value of $\ddot{a}$ we introduce the deceleration parameter
\begin{equation}
	q_0=-\frac{a\ddot{a}}{\dot{a}^2}\label{deceleration-parameter}\, ,
\end{equation}
which is an invariant quantity under the transformations \eqref{eq:rescaling-invariance}. For any given pair $a(0)$ and $\dot{a}(0)\neq0$ the deceleration parameter $q$ sets the initial conditions for $\ddot{a}(0)$. In our numerical studies we mostly use $q_0=-0.538$ from $\Lambda$CDM cosmology \cite{pdg-data} (cf.\ also the discussion below). However, we emphasize that we also view this value merely as a physically realistic choice.

Finally, as mentioned before, for a given triple $(a(0),\dot a(0),\ddot a (0))^\top=(a_0,a_1,a_2)^\top$ we solve \eqref{eq:dyn-sys-en-eq-decoupled} for a consistent value of $a_3=a^{(3)}(0)$. Unless $\dot{a}(0)=0$ the solution for $a_3$ is unique. 

As, in the end, we want to compare our equation's solutions to the $\Lambda$CDM model, we want to shortly (and partially) discuss its derivation. Mainly, this model is based on certain observation on special solutions to the Friedmann equations. These, in turn, are derived from the Einstein equation $G _\munu=8\pi G~ T_\munu$ with the assumption of a cosmological metric \eqref{eq:cosmological-metric}. Moreover, one imposes  the stress-energy tensor to be of the same homogeneity and isotropy type as the metric, that is, to take the form of a so-called perfect-fluid stress-energy tensor 
\begin{equation}
	({T^\mu}_\nu)=\textup{diag}(-\varrho,p,p,p)\label{perfectfluidSET}
\end{equation}
with functions $\varrho$ and $p$, called the energy density and the pressure, respectively. The resulting equations bear special solutions, namely by imposing the state equation $p=\gamma\cdot\varrho$ we obtain

\medskip
\begin{equation}
\text{\begin{tabular}{llll}
     $\circ$&the radiation solution & $a(t)\propto (t-t_\textup{BB})^{\nicefrac{1}{2}}$&with $\gamma=\tfrac{1}{3}$,\\[2pt]
     $\circ$&the dust solution & $a(t)\propto (t-t_\textup{BB})^{\nicefrac{2}{3}}$&with $\gamma=0$ \quad and the\hspace{2cm}~\\[2pt]
     $\circ$&the Dark Energy solution & $a(t)\propto\exp(Ht)$ &with $\gamma=-1$
\end{tabular}}
\label{eq:Friedmann-solutions}
\end{equation} 

\medskip
\noindent (for some Big Bang times $t_\textup{BB}$ and some Hubble rate $H$). For these three classes of solutions we, moreover, observe that $\varrho\propto\frac{1}{a^4},~\varrho\propto\frac{1}{a^3}$ and $\varrho=\textup{const.}$, respectively- Finally, the $\Lambda$CDM model is obtained by assuming $\varrho$ to be a superposition of these three types of energy content. Formally, we make the ansatz $\varrho=\varrho_0\big(\frac{\raisebox{1pt}{$\scriptstyle\Omega_\textup{rad}$}}{a^4}+\frac{\raisebox{1pt}{$\scriptstyle\Omega_\textup{dust}$}}{a^3}+\Omega_\textup{DE}\big)$ and obtain the equation
\begin{equation}
H^2=H_0^2\Big(~\frac{\Omega_\textup{rad}}{a^4}+\frac{\Omega_\textup{dust}}{a^3}+\Omega_{\raisebox{-1pt}{\scriptsize DE}}\Big)\label{lcdm_equation}
\end{equation}
as a cosmological model, where $H=H(t)=\frac{\dot{a}(t)}{a(t)}$ is the Hubble rate at time $t$. Hereby, $\varrho_0,~H_0,~\Omega_\textup{rad},~\Omega_\textup{dust}$ and $\Omega_{\raisebox{-1pt}{\scriptsize DE}}$ are some (not necessarily independent) parameters of the model, in particular, the latter three fulfill $\Omega_\textup{rad}+\Omega_\textup{dust}+\Omega_{\raisebox{-1pt}{\scriptsize DE}}=1$. Whenever we speak of `standard values' for these parameters we mean the values 
\[
    \Omega_\textup{rad}=5.38\cdot10^{-5},\quad\Omega_\textup{dust}=0.315,\quad\Omega_\textup{DE}=0.685
\]
and $H_0$ as above, taken from \cite{pdg-data}. Note that these values are subject to measurement uncertainty. We denote the resulting solution by $\alcdm$. 

However, motivated by the $\Lambda$CDM model's derivation we want to introduce another quantity which will frequently find use in our later discussions. Define
\[
	\Gamma[\,a\,](t):=-\frac{1}{3}\Big(2\frac{a(t)\,\ddot{a}(t)}{\dot{a}(t)^2}+1\Big)=\tfrac{2}{3}\,q[\,a\,](t)-\tfrac{1}{3},
\]
for sufficiently nice (particularly with $\dot{a}\neq0$) scale factors $a:I\to(0,\infty)$ on some interval $I$. In the last equality we used the notation $q[\,a\,]=-\frac{a\ddot{a}}{\dot{a}^2}$, obviously inspired by \eqref{deceleration-parameter}. Moreover, we note that by choosing the above parameter values for $\Omega_\textup{rad},~\Omega_\textup{dust}$ and $\Omega_{\raisebox{-1pt}{\scriptsize DE}}$ as well as the $\Lambda$CDM equation \eqref{lcdm_equation}, one can reproduce the value $q[a_{\Lambda\textup{CDM}}](0)=-0.538$ we have introduced above.

Note that $\Gamma$ has an interesting physical content. For the solutions of the Friedmann equations mentioned above, $\Gamma$ reproduces the corresponding values of $\gamma$ and, conversely, if we read the conditions $\Gamma[\,a\,]\in\big\{-1,0,\frac{1}{3}\big\}$ as ODE's in their own right, we reproduce the corresponding Friedmann solutions form \eqref{eq:Friedmann-solutions} and only these. Observing the existence of two more solutions to the $\Lambda$CDM model, namely
\[
    a(t)\propto \sinh(\beta t)^{\nicefrac{1}{2}}\qquad\textup{and}\qquad a(t)\propto \sinh(\beta t)^{\nicefrac{2}{3}}
\]
for $\Omega_\textup{dust}=0$ and for $\Omega_\textup{rad}=0$, respectively, we observe that these solutions interpolate between a radiation- or dust-like behavior at early times and a Dark Energy-like behavior at late times. $\Gamma$ in these cases reads as
\[
    \Gamma\big[\sinh(\beta t)^{\nicefrac{1}{2}}\big]=\tfrac{1}{3}-\tfrac{4}{3}\tanh(\beta t)^2\qquad\textup{and}\qquad \Gamma\big[\sinh(\beta t)^{\nicefrac{2}{3}}\big]=-\tanh(\beta t)^2,
\]
respectively, and thus, physically spoken, $\Gamma$ shows how much a given universe is radiation/dust dominated or Dark Energy dominated at a certain phase.

As a final comment on $\Gamma$, note that for any stress-energy tensor of the shape \eqref{perfectfluidSET} the corresponding (classical or semiclassical) Einstein equation immediately implies\footnote{As we defined $\Gamma$, it is nothing but the fraction of the Einstein tensor's respective diagonal entries for a FLRW metric.} that $\Gamma=\frac{p}{\varrho}$. Particularly, also the solutions of our trace equation \eqref{eq:dyn-sys-trace-eq-decoupled} may be assigned with an energy content of the Friedmann solutions' types, allowing a physical interpretation.

For our numerical simulations we want to exploit the invariance under \eqref{eq:rescaling-invariance}. To avoid numerical instability, we rescale with $\beta_2=\frac{1}{H_0}$ (and, correspondingly, redefine $k_4$) and end up with the initial value $\dot{a}(0)=a_1=1$ in the new time scale. $q_0$ is not affected by our rescaling and $a^{(3)}$ is still computed by the energy equation, now with $k_4$ in the new time scale. $\lambda_0$ is usually set to 1, as a different value may be absorbed into the renormalization freedom. We employ the standard stiff\footnote{We are particularly interested in Big Bang solutions.} equation solver $\texttt{ode15s}$ of the \textsc{Matlab}\textsuperscript{\textregistered} R2020a release\footnote{A comparison to other solvers showed little to no deviation between solutions, with deviations decreasing as the solvers' accuracies were increased.}. { Note that the numerical solver does not integrate into the $a=0$-singularity of \eqref{eq:dyn-sys-trace-eq-decoupled} in a strict sense, but stops at $a$-values $\approx 10^{-7}$ to $10^{-9}$. Particularly, we do not make any claims on Planck scale cosmology, where the validity of the SCE is expected to break down due to quantum effects of the space-time itself.}

\begin{figure}
\centering
\begin{subfigure}{.52\textwidth}
\centering
\includegraphics[scale=1]{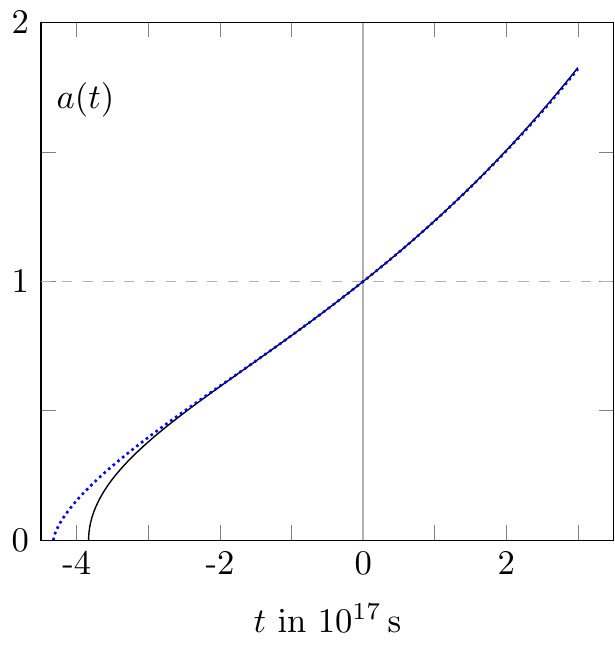}
\end{subfigure}
\begin{subfigure}{.46\textwidth}
\centering
\includegraphics[scale=1]{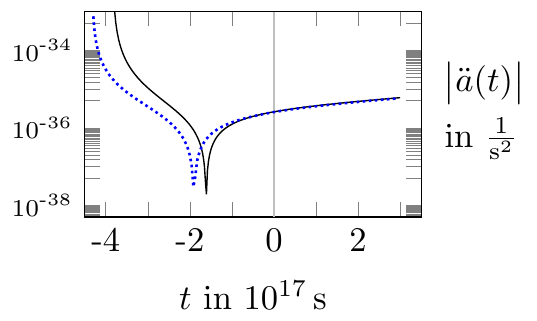}
\\
\,\includegraphics[scale=1]{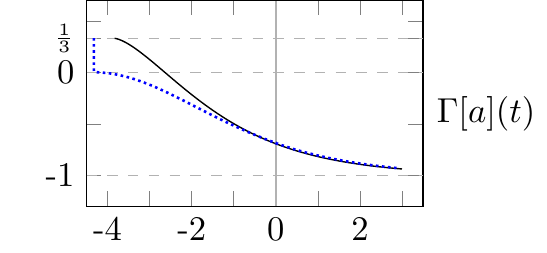}

\vspace{.65cm}
\end{subfigure}
\begin{minipage}{.9\textwidth}
\captionsetup{format=plain, labelfont=bf}
\caption{Solution of the cSCE (black), compared to the best fit (i.e.\ using numeric values from \cite{pdg-data}, cited in the text) $\Lambda$CDM model (dotted blue)\label{firstsolution}}
\end{minipage}
\end{figure}

A typical solution with the parameters

\vspace{.2cm}
\begin{tabular}{p{1.7cm}p{4cm}p{2.1cm}p{2.1cm}}
\multicolumn{2}{p{4.7cm}}{$\circ$~ $\dot{a}(0)=H_0=2.2\cdot10^{-18}\frac{1}{\textup{s}}$}&\multicolumn{2}{l}{$\circ$~ $\ddot{a}(0)$ given by $q_0=q_{0,\Lambda\textup{CDM}}=-0.538$}\\[5pt]
$\circ$~ $\xi=\frac{1}{12}$&$\circ$~ $3c_3+c_4=0.5$&$\circ$~ $\lambda_0=1$&$\circ$~ $\kappa=2\cdot10^{42}$
\end{tabular}
\vspace{.2cm}\\
is shown in Figure \ref{firstsolution}. It is typical in the sense that its behavior as a function of time is generic for a certain range of parameters that has been identified manually in order to retrieve promising cosmological models. The most remarkable of these properties are an exponential late time expansion as well as a `square-root-like zero' at early times. In other words, we indeed observe a solution which admits a Big Bang and immediately after this Big Bang the universe expands asymptotically as $a(t)\propto(t-t_\textup{BB})^{\nicefrac{1}{2}}$. To underpin this notion of `square-root-like', we have included a plot of the solution $a$ of Figure \ref{firstsolution} $-$ together with its first two derivatives $-$ in Figure \ref{firstsolutionloglog}. The horizontal axes in Figure \ref{firstsolutionloglog} are now shifted to $t-t_\textup{BB}$ (with a numerically obtained value $t_\textup{BB}$), allowing a log-log-scaling. The red dotted lines show the analog curves for a pure square-root expansion with the same Big Bang time and, particularly, how the latter fits our solution over several magnitudes. Moreover, computing the Ricci scalar curvature for the metric \eqref{eq:cosmological-metric}, that is $R=6\big(\frac{\ddot{a}}{a}+\frac{\dot{a}^2}{a^2}\big)$, we find that our Big Bang is indeed a singularity in the sense that $R\to\infty$ as $t\to t_\textup{BB}$. 

In terms of the quantity $\Gamma$, the above observations can be  interpreted as a radiation dominated early phase and a Dark Energy dominated late time expansion. The former does match the physical expectation that a massless scalar field should behave like radiation, and the latter does again indicate an effect of Dark Energy, although we did not include a cosmological constant to our model. Note that one cannot easily blame a non-vanishing cosmological constant for this effect, since the influence of $c_1$ (that is, the renormalization constant of $\Lambda$) is ruled out by setting $m=0$.

\begin{remark}
\normalfont
Note that for metrics of the form \eqref{eq:cosmological-metric} one can compute
\[
	\tfrac{1}{6}\,\Box\, R=\frac{a^{(4)}}{a}+3\frac{\dot{a}a^{(3)}}{a^2}+\frac{\ddot{a}^2}{a^2}-5\frac{\ddot{a}\dot{a}^2}{a^3},
\]
that is, the first line of \eqref{eq:dyn-sys-trace-eq-decoupled} is proportional to $\square R$. Hence, for parameters 
$\varepsilon:=3c_3+c_4,\xi$ and $\kappa$ $($as well as $\lambda_0)$ such that $k_1\gg k_2,k_3,k_4$, the trace equation is expected to be well approximated by 
\begin{equation}
	\Box \, R=0,\label{boxRisequaltozero} 
\end{equation}
at least sufficiently far away from the singularity defined by $k_2\log(\lambda_0a)-k_1=0$. Note that \eqref{boxRisequaltozero} is also solved  by $a(t)\propto(t-t_\textup{BB})^{\nicefrac{1}{2}}$, by $a(t)\propto\exp(Ht)$ and by $a(t)\propto\sinh(t-t_\textup{BB})^{\nicefrac{1}{2}}$, which in turn solve the $\Lambda$CDM model for particular choices of matter. 
\end{remark}

\begin{figure}
\centering
\begin{subfigure}{.48\textwidth}
\centering
\includegraphics[scale=1]{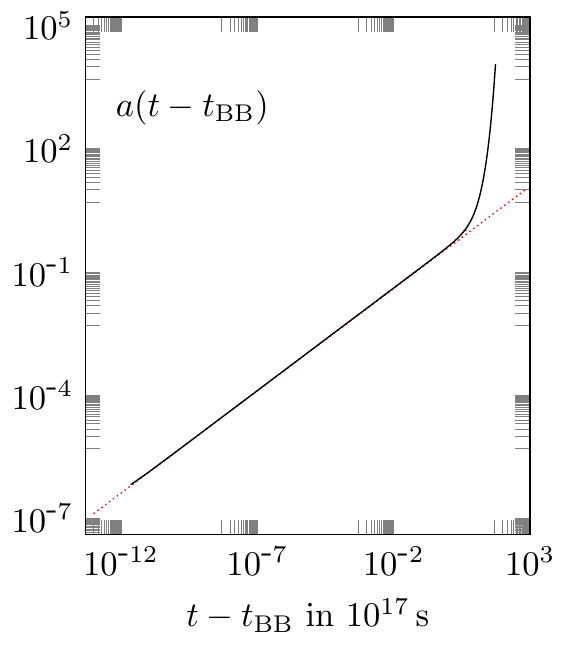}
\end{subfigure}
\begin{subfigure}{.51\textwidth}
\centering
\includegraphics[scale=1]{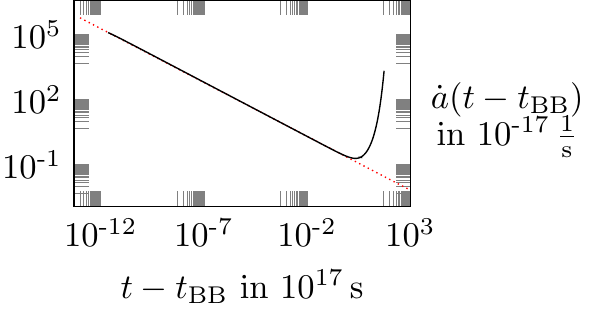}\,\,\,
\\
\includegraphics[scale=1]{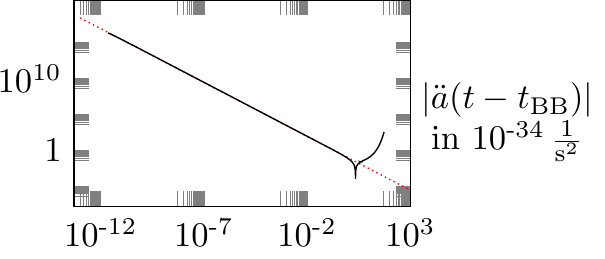}
\vspace{.6cm}
\end{subfigure}
\begin{minipage}{.9\textwidth}
\captionsetup{format=plain, labelfont=bf}
\caption{Double logarithmic plot of the solution in Figure \ref{firstsolution} (black), together with the regression lines whose slopes coincide with a square-root function (dotted red)\label{firstsolutionloglog}}
\end{minipage}
\end{figure}

\section{Numeric solutions of the cosmological semiclassical Einstein equation}
\label{sec:Numeric-solutions-of-the-cosmological-semiclassical-Einstein-equation}


In the present section we identify a few regions of interest in the parameter space of our cosmological model.

Throughout this section we will denote $\varepsilon=3c_2+c_3$ and usually we set $\lambda_0=1$. Moreover, we denote by $\varepsilon_\textup{crit}=\varepsilon_\textup{crit}(\xi)$ the value such that $k_1=k_1(\varepsilon,\xi)$ vanishes. Recall that varying $\dot{a}(0)$ does not influence the shape of our solutions and may be regarded as a redefinition of a time scale (while accordingly adjusting $\kappa$'s units). Thus, we generally omit an investigation of the dependency of our solutions on $\lambda_0$ and $\dot{a}(0)$. 


\subsection{Influence of the renormalization freedom and the curvature coupling}
\label{sec:Influence_ren_freedom_and_curv_coupling}

We start the numerical investigations with the parameter dependency of solutions of the generic type shown in Figures \ref{firstsolution} and \ref{firstsolutionloglog} and present a family of solutions in Figure \ref{eps-influence}. In particular, this includes a few more numerical solutions at $\xi=\frac{1}{12}$ which we count to the generic class.

The choice of values of $\varepsilon$ was made to show the behavior of solutions around the critical value $\varepsilon_\textup{crit}(\frac{1}{12}=-\frac{1}{960\pi^2}\approx-1.0554\cdot10^{\textup{-}4}$. Far remote from this value solutions are captured by the red curve (a) in Figure \ref{eps-influence}. For this critical value we have $k_1=0$ as well as $a_\textup{crit}=1$ and we cannot solve our trace equation for $a^{(4)}$ at our choice of initial values. Close to that value we observe unstable behavior. The value $\xi=\frac{1}{12}$,  where there exists no pure de Sitter solution  (cf.\ Figure \ref{fig:HdS_of_xi}), was chosen as an example for cases with the aforementioned property. Different choices for $\xi$ with $|\xi-\frac{1}{6}|\ge\sqrt{\nicefrac{1}{4320}}$ produce similar graphics.

\begin{figure}[t!]
\centering
\includegraphics[scale=1]{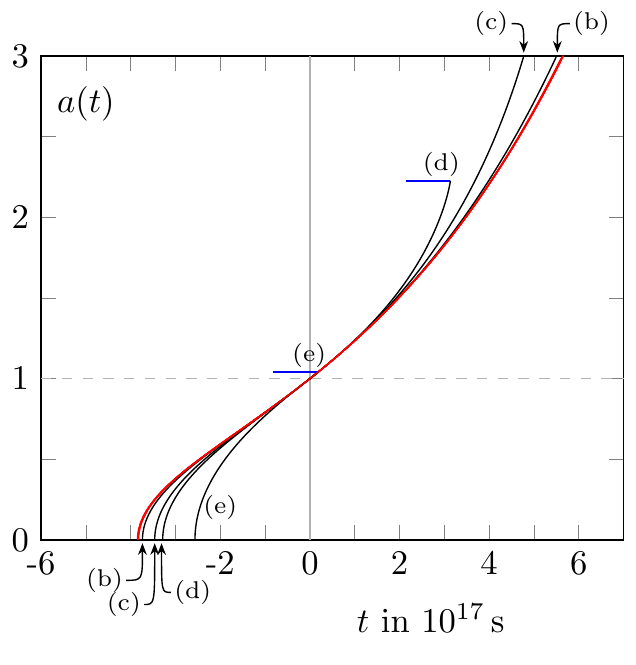}
\qquad\qquad
\includegraphics[scale=1]{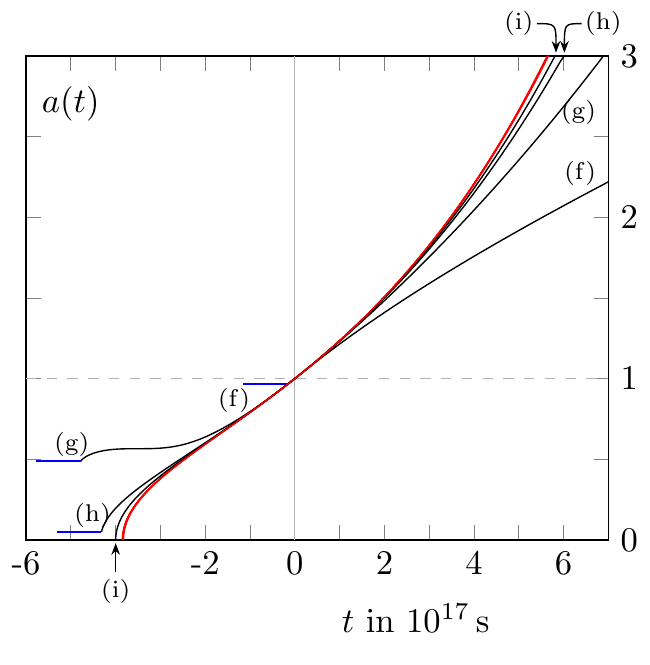}

\begin{minipage}{.41\textwidth}
\captionsetup{format=plain, labelfont=bf}
\caption{Influence of the parameter $\varepsilon=3c_3+c_4$ shown in a family of solutions. The cases `$\ge10^{\textup{-}2}$' and `$\le-10^{\textup{-}2}$', labeled by (a), contain the values $\{-100,-1,-\frac{1}{10},-10^{\textup{-}2}\}$ and $\{10^\textup{{-}2},\frac{1}{10},\frac{1}{2},100\}$, respectively, yielding a critical value for $a$ of numerical zero or numerical infinity. All these solutions show no significant difference among each other and are covered by the red curve. The blue lines mark the respective critical value of $a_\textup{crit}=\exp\big(\frac{k_1}{k_2}\big)$ for the other cases.\label{eps-influence}} 
\end{minipage}\hfill
\begin{minipage}{.57\textwidth}
\begin{tabular}{rclcrclcrcl}
\multicolumn{3}{l}{Parameters:}\\[2pt]
$\dot{a}(0)$\hspace{-.2cm}&\hspace{-.2cm}$=$\hspace{-.2cm}&\hspace{-.2cm}$H_0$,&&
$q_0$\hspace{-.2cm}&\hspace{-.2cm}$=$\hspace{-.2cm}&\hspace{-.2cm}$-0.538$,&&
$\kappa$\hspace{-.2cm}&\hspace{-.2cm}$=$\hspace{-.2cm}&\hspace{-.2cm}$2\cdot10^{42}$,\\
$\lambda_0$\hspace{-.2cm}&\hspace{-.2cm}$=$\hspace{-.2cm}&\hspace{-.2cm}$1$,&&
$\xi$\hspace{-.2cm}&\hspace{-.2cm}$=$\hspace{-.2cm}&\hspace{-.2cm}$\frac{1}{12}$~~~~~&
\end{tabular}
\tabulinesep=2pt

\medskip
~~\begin{tabular}{crccrc}
~&$\varepsilon$~~~~&$\exp\big(\frac{k_1}{k_2}\big){\color{white}\Big|}$~&&$\varepsilon$~~~~&$\exp\big(\frac{k_1}{k_2}\big){\color{white}\Big|}$\\\hline
{\color{red}(a)}&\hspace{-.2cm}$\ge10^{\textup{-}2}$&$\approx0${\color{white}\LARGE A\!\!\!}&(f)&\hspace{-.2cm}$\textup{-}1.1\!\cdot\! 10^{\textup{-}4}$&0.9668
\\
(b)&$10^{\textup{-}3}$&4358.4&(g)&$\textup{-}2\!\cdot\! 10^{\textup{-}4}$&0.4887\\
(c)&$10^{\textup{-}4}$&4.7492&(h)&$\textup{-}5\!\cdot\! 10^{\textup{-}4}$&0.0503\\
(d)&$0$~~~&2.2255&(i)&$\textup{-}10^{\textup{-}3}$&0.0011\\
(e)&$\textup{-}10^{\textup{-}4}$&1.0429&{\color{red}(a)}&$\le\textup{-}10^\textup{{-}2}$&$\approx\infty$
\end{tabular}
\end{minipage}

\end{figure}

Note that any solution exists until it runs into one of the singularities $a=0$ or $a=a_\textup{crit}$. Around the aforementioned instability at $\varepsilon=\varepsilon_\textup{crit}(\frac{1}{12})$ we observe that $a_\textup{crit}$ approaches the value 1 and, accordingly, we end up with a short interval of existence. 

\begin{remark}
\normalfont
We want to emphasize that for any numerical solution we have observed to run into the $a_\textup{crit}$-singularity the values of $\dot{a}$ apparently remain finite in that limit. This is not very surprising since the vector field we integrate for the solution has a pole of order one at $a_\textup{crit}$. Hence, a sloppy analysis suggests that $a$ is a function whose fourth derivative has a pole of order one, implying that its third derivative has a logarithmic pole and that its second and first derivatives as well as $a$ itself can be continuously extended to that critical point and beyond.

Note that the immediate output of our numerical solver, which returns $a$ and its first three derivatives, shows that $\dot{a}$ diverges in such points.  Plugging the solvers output into the energy constraint's RHS and recalling the discussion from the beginning of Section \textup{\ref{sec:General_discussion_on_the_decoupled_cSCE}}, however, suggest this divergence to be a numerical artifact.
\end{remark} 

The unstable behavior for $\varepsilon\to\varepsilon_\textup{crit}(\frac{1}{12})=-\frac{1}{960\pi^2}$  can now be characterized as follows. As $\varepsilon\to\varepsilon_\textup{crit}(\xi)$ we have $a_\textup{crit}\to 1$. Hence, on the one hand, if $\varepsilon>\varepsilon_\textup{crit}(\frac{1}{12})$ (left graphic in figure \ref{eps-influence}) we have an existence interval of the form $(t_\textup{BB},\eta)$ with some $\eta=\eta(\varepsilon)>0$ and some $t_\textup{BB}=t_\textup{BB}(\varepsilon)<0$, where in particular
\[
	\eta\to0\quad\textup{as}\quad\varepsilon\to\varepsilon_\textup{crit}(\tfrac{1}{12})\hspace{1cm}\textup{and}\hspace{1cm} \eta\to \infty\quad\textup{as}\quad\varepsilon\to+\infty.	
\]
Moreover, we observe that
\[
t_\textup{BB}\to t_\textup{BB,eff}\quad\textup{as}\quad\varepsilon\to\varepsilon_\textup{crit}(\tfrac{1}{12})\hspace{1cm}
\]
with some $t_\textup{BB,eff}<0$, playing the role of an effective Big Bang time in the limit.
On the other hand, if $\varepsilon<\varepsilon_\textup{crit}(\frac{1}{12})$ (right graphic in figure \ref{eps-influence}), the solution exist on an interval of the form $(-\eta,\infty)$ for some $\eta>0$, now with 
\[
\eta\to0\quad\textup{as}\quad\varepsilon\to\varepsilon_\textup{crit}(\tfrac{1}{12})\hspace{1cm}\textup{and}\hspace{1cm}\eta\to -t_\textup{BB,(a)}\quad\textup{as}\quad\varepsilon\to-\infty,
\]
where $t_\textup{BB,(a)}<0$ is the Big Bang time of the limit curve (a).

If we want to combine the two resulting limits 
 for $\varepsilon\to\varepsilon_\textup{crit}(\frac{1}{12})$, defined on the intervals $(t_\textup{BB,eff},0)$ and $(0,\infty)$, respectively, our numerical analysis suggests that  we obtain a square-root power law expansion. This is already indicated in curves (e) and (f) in Figure \ref{eps-influence} (or their respective branch) and behavior becomes more pronounced, if we choose values of $\varepsilon$ even closer to $\varepsilon_\textup{crit}(\frac{1}{12})$. 

Finally, 
the term $k_1\Box R$ in the trace equation, which originates in pure quantum effects, usually induces solutions with an exponential late-time expansion as remarked in Section \ref{sec:Short-derivation-of-sce}. 

\begin{remark}
\normalfont
Recall R.\ M.\ Wald's classical work \textup{\cite{Wald_deSitter}}, where he shows that solutions to the $($classical$)$ cosmological Einstein equation with a positive cosmological constant usually $($i.e.\ under some assumption on the stress-energy tensor$)$ show a late time exponential expansion. Thus, it is noteworthy that in our case  the  $k_1\Box R$ term seemingly plays a similar role as the classical cosmological constant. 
Tuning the prefactor of $\Box R$ to zero, we apparently restore a purely radiation dominated expansion.
\end{remark}

\begin{figure}
    \centering
    \includegraphics[scale=1]{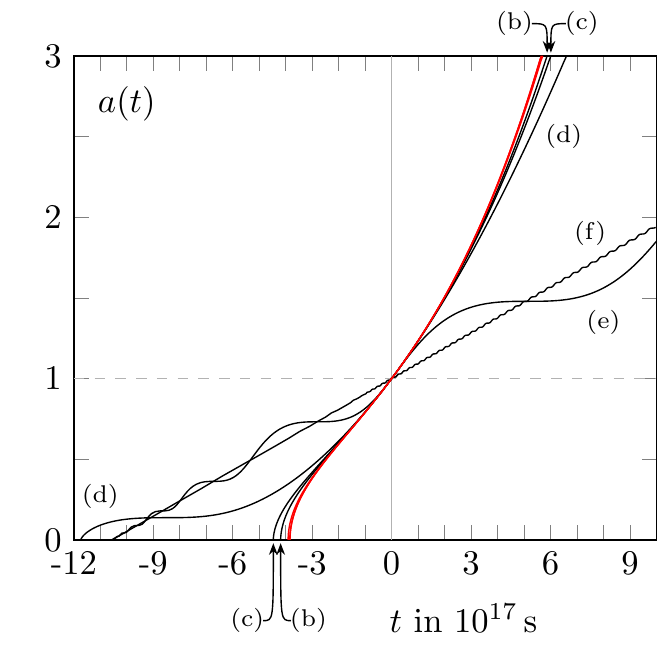}
    \qquad\qquad
    \includegraphics[scale=1]{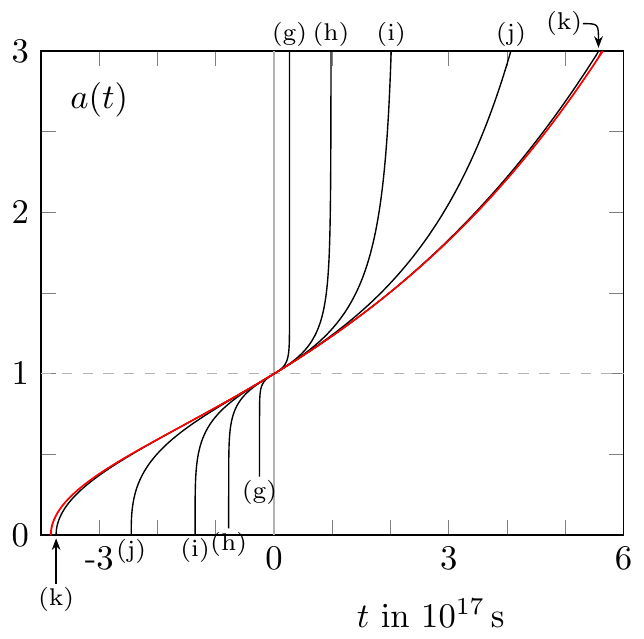}
    \begin{minipage}{.41\textwidth}
\captionsetup{format=plain, labelfont=bf}
\caption{Solutions of the trace equation with the listed parameters. Note that the critical value of renormalization constants is given by $\varepsilon_\textup{crit}=-\frac{1}{5760\pi^2}\approx-1.75905\cdot10^{\textup{-}5}$, in between the respective parameters of curve (f) and (g). Solutions for values above the latter are plotted on the left, for values below on the right. \label{conformally-coupled-case}}
\end{minipage}\hfill
\begin{minipage}{.57\textwidth}
\begin{tabular}{rclcrclcrcl}
\multicolumn{3}{l}{Parameters:}\\[2pt]
$\dot{a}(0)$\hspace{-.2cm}&\hspace{-.2cm}$=$\hspace{-.2cm}&\hspace{-.2cm}$H_0$,&&
$q_0$\hspace{-.2cm}&\hspace{-.2cm}$=$\hspace{-.2cm}&\hspace{-.2cm}$-0.538$,&&
$\kappa$\hspace{-.2cm}&\hspace{-.2cm}$=$\hspace{-.2cm}&\hspace{-.2cm}$2\cdot10^{42}$,\\
$\lambda_0$\hspace{-.2cm}&\hspace{-.2cm}$=$\hspace{-.2cm}&\hspace{-.2cm}$1$,&&
$\xi$\hspace{-.2cm}&\hspace{-.2cm}$=$\hspace{-.2cm}&\hspace{-.2cm}$\frac{1}{6}$&
\end{tabular}
\tabulinesep=2pt

\medskip
\quad\begin{tabular}{crccr}
&$\varepsilon$\quad~&~~~&&$\varepsilon$\quad~\\\hline
{\color{red}(a)}&{\color{white}\LARGE A}$\ge10^{\textup{-}4}$&&(g)&$\textup{-}1.76\!\cdot\!10^{\textup{-}5}$\\
(b)&$10^{\textup{-}5}$&&(h)&$\textup{-}1.77\!\cdot\!10^{\textup{-}5}$\\
(c)&$0$&&(i)&$\textup{-}1.78\!\cdot\!10^{\textup{-}5}$\\
(d)&$\textup{-}10^{\textup{-}5}$&&(j)&$\textup{-}1.8\!\cdot\!10^{\textup{-}5}$\\
(e)&$\textup{-}1.7\!\cdot\!10^{\textup{-}5}$&&(k)&$\textup{-}2\!\cdot\!10^{\textup{-}5}$\\
(f)&$\textup{-}1.759\!\cdot\!10^{\textup{-}5}$&&{\color{red}(a)}&$\le\textup{-}10^{\textup{-}3}$
\end{tabular}
\end{minipage}
\end{figure}

\begin{figure}[t!]
\centering
\begin{minipage}{.42\textwidth}
\centering
\includegraphics[scale=1]{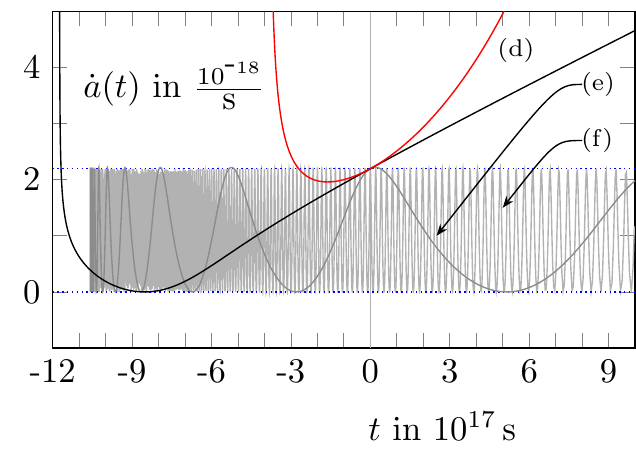}\quad
\end{minipage}\qquad
\begin{minipage}{.42\textwidth}
\includegraphics[scale=1]{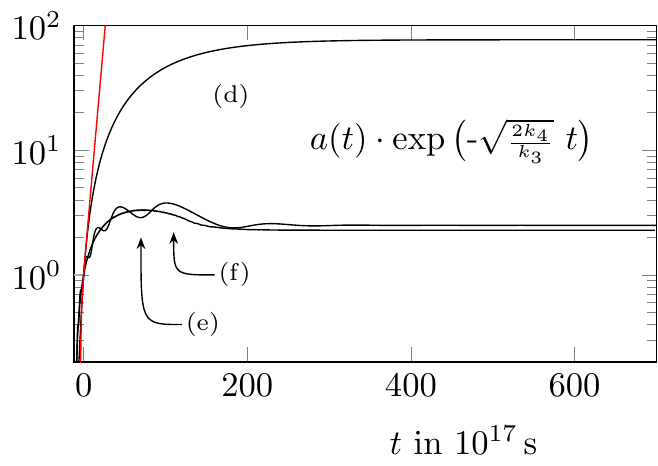}
\end{minipage}

\begin{minipage}{.42\textwidth}
\captionsetup{format=plain, labelfont=bf}
\caption{The graphics show the quantities $\dot{a}$, $a/a_\textup{dS}$ and $\Gamma[\,a\,]$ deduced from $a$ for the parameter settings of Curves (d), (e) and (f) of Figure \ref{conformally-coupled-case}, that is, for $\varepsilon$ approaching $\varepsilon_\textup{crit}$ from above. For reference the red curve shows the `generic' curve with $\varepsilon$ far remote from $\varepsilon_\textup{crit}$, labelled (a) in Figure \ref{conformally-coupled-case}.
\label{conformally-coupled-case-two}}
\end{minipage}\qquad
\begin{minipage}{.42\textwidth}
\centering
\includegraphics[scale=1]{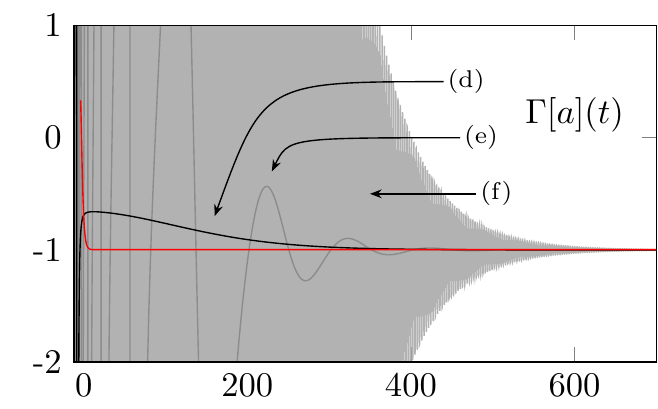}
\end{minipage}
\end{figure}

\bigskip
As a second part of this section, we want to discuss the influence of the curvature coupling $\xi$ by studying a family of solutions with $\varepsilon\to\varepsilon_\textup{crit}$ for another value of $\xi$. From a physical point of view, we have the distinguished cases $\xi=0$, called minimally coupled case, and $\xi=\frac{1}{6}$, called conformally coupled case. Formally, the minimally coupled case plays no particularly distingished role. 


As we have discussed before,  in the case $\xi=\frac{1}{6}$ the trace equation simplifies significantly. Particularly, we have $k_2=0$, which cancels many terms of the trace equation including the $\log(a)$-term. Consequently, the influence of the $k_1\,\Box\, R$-term is controlled by one parameter in a simple manner, namely $\varepsilon$, and not by a possibly singular dependency on the value of $a$ (such as our $\log(a)$-term). Particularly, we have no critical value of $a$ in the limit $k_1\to0$ (i.e.\ in the limit $\varepsilon\to\varepsilon_\textup{crit}(\frac{1}{6})=-\frac{1}{5760\pi^2}$). For $k_1=0$ we are thus back in the Starobinski scenario from Section \ref{sec:General_discussion_on_the_decoupled_cSCE}. Recall that for $\xi=\frac{1}{6}$ the solution of a pure de Sitter expansion with $H^\textup{dS}$ from \eqref{eq:deSitterSolutions} are present.

Some numerical solutions for $\xi=\frac{1}{6}$ are shown in Figure \ref{conformally-coupled-case}. Recall that in Figure \ref{eps-influence} the behavior is different if we approached the critical value of $\varepsilon$ from below or above. But the limiting curves appeared to be, in a suitable sense, consistent. In Figure \ref{conformally-coupled-case}, in turn, this is not the case anymore. As expected, for a large absolute value of $\varepsilon$ we recover the generic solution class as before. 

If we approach $\varepsilon\to\varepsilon_\textup{crit}=-\frac{1}{5760\pi^2}$ from above we can see how an oscillating behavior fades in. 
The oscillation's frequency grows as $\varepsilon\to\varepsilon_\textup{crit}$ and decays as $t\to\infty$. The amplitude, in turn, decays as $\varepsilon\to\varepsilon_\textup{crit}$ but appears to grow as $t\to\infty$. Solutions decay in steps and the slope of these steps is bounded from above by a value slightly larger than the initial value $\dot{a}(0)=H_0$ as well as from below by 0. 

To continue the analysis we have included plots of the quantities $\dot{a}$, $a/a_\textup{dS}$ and $\Gamma[\,a\,]$ deduced from the solution $a$ in Figure \ref{conformally-coupled-case-two} for Curves (d), (e) and (f). There we observe again the aforementioned claims on $\dot{a}$, in particular the (approximate) boundary interval $[0,H_0]$ for $\dot{a}$ is indicated by the blue dotted line. In the limit $\varepsilon\to\varepsilon_\textup{crit}$ from above, our solutions in Figure \ref{conformally-coupled-case} seemingly converge to a linear expansion. This, however, is no longer true on a larger time scale. In Figure \ref{conformally-coupled-case-two} this is shown by rescaling the solutions with the pure de Sitter expansion, that is, by plotting $a/a_\textup{dS}$ with $a_\textup{dS}(t)=\exp(H^\textup{dS}t)$ with $H^\textup{dS}$ from \eqref{eq:deSitterSolutions}.  Note that all solutions for sufficiently small $k_1>0$ result in an exponential late-time expansion with de Sitter rate $H^\textup{dS}$. The latter value can be reproduced by solving the `limit equation'
\[
	a^{(4)}=-\frac{k_3}{k_1}\,\frac{\dot{a}^2\ddot{a}}{a^2}+\frac{k_4}{k_1}\Big(\frac{\dot{a}^2}{a}+\ddot{a}\Big)\qquad\textup{with}\quad a(t)\propto\exp(\widetilde{H}^\textup{dS}t)\quad\textup{and}\quad \widetilde{H}^\textup{dS}=\sqrt{\frac{2k_4}{k_1+k_3}},
\]
where the `limit' hereby refers to, after having solved \eqref{eq:dyn-sys-trace-eq-decoupled} for $a^{(4)}$, neglecting all terms which do not scale by $\frac{1}{k_1}$. It is noteworthy that $\widetilde{H}^\textup{dS}\to H^\textup{dS}$ as $k_1\to0$, that is, the latter limit recovers the $\xi=\frac{1}{6}$\,-value of $H^\textup{dS}$.

The emergence of late-time de Sitter expansions can, moreover, be observed in the $\Gamma[\,a\,]$-plots in Figure \ref{conformally-coupled-case-two}, where at late times each solution yields a Dark Energy dominated universe with $\Gamma[\,a\,](t)\to-1$ as $t\to\infty$. Note that $\Gamma[\,a\,]$ appears as approaching its limit $-1$ similarly to how a damped harmonic oscillator reaches its stable equilibrium. Note that the `generic' solutions with large $k_1$ (the red curves in Figures \ref{conformally-coupled-case} and \ref{conformally-coupled-case-two} labeled as Curve (a)) end in a Dark Energy-dominated late-time expansion as well. However, for sufficiently large $k_1$ the effective late-time de Sitter coefficient differs from the value $H^\textup{dS}$.

If, on the other hand, we approach $\varepsilon\to\varepsilon_\textup{crit}$ from below, the solutions tend  to 0 for $t<0$ and to infinity 
for $t>0$ on decreasingly short time scales. These solutions are shown in the right graphic of Figure \ref{conformally-coupled-case}. However, from a cosmological viewpoint these solutions do not appear particularly useful.

To close the discussion of Figures \ref{eps-influence} and \ref{conformally-coupled-case} (together with \ref{conformally-coupled-case-two}), we remark that similar graphics can be generated for any value of $\xi$. On the one hand, Figure \ref{eps-influence} is representative for values with $|\xi-\frac{1}{6}|\ge\sqrt{\nicefrac{1}{4320}}$ (i.e.\ outside the interval distinguished in Section \ref{sec:General_discussion_on_the_decoupled_cSCE}). On the other hand, for values with $0<|\xi-\frac{1}{6}|<\sqrt{\nicefrac{1}{4320}}$ the effects of Figure \ref{eps-influence} (particularly, the influence of a positive $a_\textup{crit}$) and the effects of Figure \ref{conformally-coupled-case} (particularly, the presence of an attractive de Sitter solution with parameter $H^\textup{dS}$) mix up, but we have not found any new behavior of solutions with $\varepsilon$ around $\varepsilon_\textup{crit}$. 

\begin{figure}[t!]
\begin{subfigure}{.495\textwidth}
\centering
\includegraphics[scale=1]{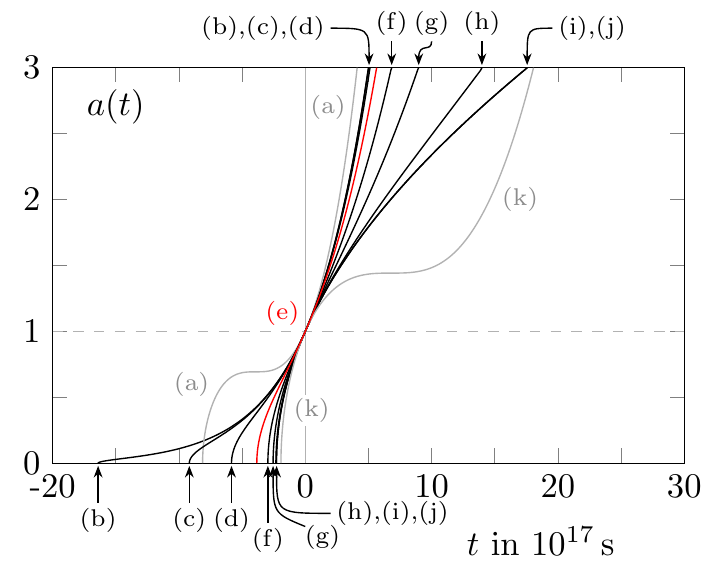}
\end{subfigure}
\begin{subfigure}{.495\textwidth}

\includegraphics[scale=1]{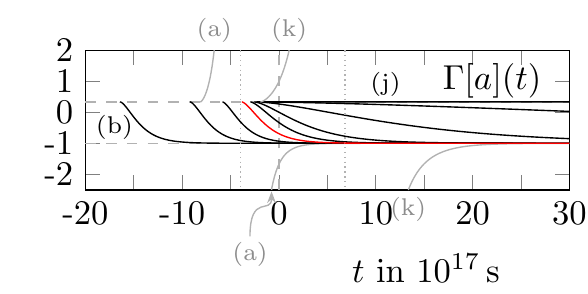}

\vspace{-.2cm}
\includegraphics[scale=1]{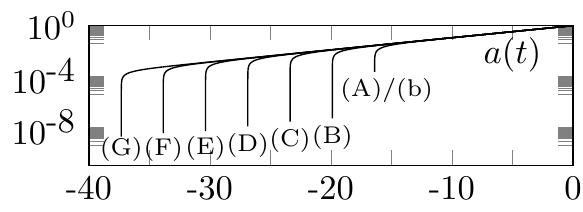}

\vspace{.5cm}
\end{subfigure}
\begin{minipage}{.45\textwidth}
\captionsetup{format=plain, labelfont=bf}
\caption{The left graphic shows several solutions with a varying deceleration para- meter $q_0\in[-1,1]$ labeled by (b) to (j). The gray curves, (a) and (k), show solutions for $q_0$ outside that interval. On the right top we show the respective plots of $\Gamma[\,a\,]$. The unlabeled curves belong to the parameters in the obvious order, that is, between (b) and (j) we have (c) to (i) from left to right. The lower right graphic shows several solutions with $q_0=-1$ in a logarithmic plot, indicating how with an increasing value of $\varepsilon$ the solutions better and better approximate a de Sitter solution, also at small values of $a$. Note that for these parameters $a_\textup{crit}$ is numerically infinite.\label{deceleration-par-fig}}
\end{minipage}
\begin{minipage}{.54\textwidth}
\begin{tabular}{rclcrcl}
\multicolumn{4}{l}{Parameters:}\\[2pt]
$\dot{a}(0)$\hspace{-.2cm}&\hspace{-.2cm}$=$\hspace{-.2cm}&\hspace{-.2cm}$H_0$&&$\kappa$\hspace{-.2cm}&\hspace{-.2cm}$=$\hspace{-.2cm}&\hspace{-.2cm}$2\cdot10^{42}$\\
$\lambda_0$\hspace{-.2cm}&\hspace{-.2cm}$=$\hspace{-.2cm}&\hspace{-.2cm}$1$&&$\xi$\hspace{-.2cm}&\hspace{-.2cm}$=$\hspace{-.2cm}&\hspace{-.2cm}$\frac{1}{12}$
\end{tabular}

\medskip
~~\begin{tabular}{llll}
\multicolumn{2}{c}{$\varepsilon=1$:~~~}&\multicolumn{2}{c}{$q_0=-1$:}\\\hline
(a)&$q_0=-2$&(A)&$\varepsilon=1$\\
(b)&$q_0=-1$&(B)&$\varepsilon=10$\\
(c)&$q_0=-0.99$&(C)&$\varepsilon=10^2$\\
(d)&$q_0=-0.9$&(D)&$\varepsilon=10^3$\\
\color{red}(e)&$q_0=-0.538$&(E)&$\varepsilon=10^4$\\
(f)&$q_0=0$&(F)&$\varepsilon=10^5$\\
(g)&$q_0=0.5$&(G)&$\varepsilon=10^6$\\
(h)&$q_0=0.9$&\\
(i)&$q_0=0.99$&\\
(j)&$q_0=1$&\\
(k)&$q_0=2$&
\end{tabular}
\end{minipage}
\end{figure}

We remark that similar observations have been made in \cite{Dappiaggi_etal_DE_from_QFT}, where the authors approximate the state's contributions to the back-reaction equation. In this different setting they also observe that the respective Starobinski solution $a(t)\propto\exp(H^\textup{dS}t)$ is attractive if $\varepsilon-\varepsilon_\textup{crit}$ has the correct sign, and is repulsive for the respective opposite sign.

The solutions shown in this section, at least for $\varepsilon>\varepsilon_\textup{crit}(\xi=\frac{1}{6})$, underpin our observation of a late-time de Sitter expansion being generic. 

\subsection{Influence of the initial values}

If we specify a certain interval of `reasonable' $q_0$-values, we again end up  with the generic solution class from Section \ref{sec:General_discussion_on_the_decoupled_cSCE}, where by `reasonable' we refer to values $q_0\in[-1,1]$, that is, such initial valued for $\ddot{a}$ for which the initial value of $\Gamma[\,a\,]$ fulfills $\Gamma[\,a\,]\in[-1,\frac{1}{3}]$.

The left graphic of Figure \ref{deceleration-par-fig} shows the transition from an (approximately) exponential expansion ($q_0=-1$, curve (b)) to a square-root-like expansion ($q_0=1$, curve (j)). Still, for $q_0=-1$ we observe a radiation-like expansion at very early times. The upper right graphic in Figure \ref{deceleration-par-fig} shows the respective curves of $\Gamma[\,a\,]$. 
The lower right graphic shows a family of solutions with a variation of $\varepsilon$ on a logarithmic scale, starting with curve (b) and increasing $\varepsilon$. 
The Dark Energy dominated period is pushed to  smaller values of $a$ by increasing $\varepsilon=3c_3+c_4$, or $k_1$, respectively. 

Curves (a) and (k) in Figure \ref{deceleration-par-fig} show solutions with values for the deceleration parameter outside the interval $[-1,1]$, namely for $q_0=-2$ and $q_0=2$. On both sides of said interval we observe an inflection point with zero derivative, at $t<0$ for $q_0<-1$ and at $t>0$ for $q_0>1$. Plotting more curves, one would, moreover, see convergence of this inflection point  to $t=0$ for both $q_0\to \infty$ and $q_0\to-\infty$ and in both these limits the solutions  converge to the same function, now with an inflection point with zero derivative at $t=0$. An inflection point with zero derivative of some $a$ does imply a divergence of $\Gamma[\,a\,]$, which we can observe in the upper right graphic of Figure \ref{deceleration-par-fig}.

\subsection{Cosmic horizon problem}
\label{sec:Horizon}
We shortly recall the definition of conformal time. For a FLRW-type space-time with scaling factor $a$, we reparameterize the time coordinate by $\tau(t)=\int_0^t\big(a(t')\big)^{\textup{-}1}~\textup{d}t'$ such that $g=a(\tau)^2(-\textup{d}\tau^2+g_{\R^3})$ holds in these new coordinates. In conformal time, a causal connection of two space-time points is given, if and only if they are causally connected in Minkowski space-time. For a Big Bang-solution $a$ with zero $t_\textup{BB}$ we define $\tau_\textup{BB}:=\tau(t_\textup{BB})$. 

The cosmic horizon problem concerns the extremely homogeneous state of the observable universe.  If, in an universe given by $a(t)$, two observable regions with the same matter distribution  are not causally connected, this would exclude a homogenizing process in the common causal past of both regions. One solution to the cosmic horizon problem is that \emph{all} observable regions of the universe have a common causal past, which is achieved    
by a large negative value of $\tau_\textup{BB}$ or even $\tau_\textup{BB}=-\infty$. This is realized by theories of the inflationary early universe \cite{Guth,Linde,Liddle}.


Here, we want to investigate how much the cosmological model introduced in the present paper is compatible with solutions to the cosmic horizon problem. We observed in Section \ref{sec:Influence_ren_freedom_and_curv_coupling} that in specific regions of the parameter space the solutions $a(t)$ show an inflection point with vanishing first derivative and we can even have arbitrarily many of them, see e.g. Figure \ref{conformally-coupled-case} with $\varepsilon>\varepsilon_\textup{crit}(\xi)$. 
Tuning $\varepsilon$ such that the inflection point coincides with the Big Bang,  we obtain $\tau_\textup{BB}=-\infty$. Note that choosing $\xi\neq\frac{1}{6}$ requires a value $\varepsilon>\varepsilon_\textup{crit}(\xi)$ in order to guarantee $a_\textup{crit}>1$, otherwise the solutions do not exist long enough to admit a Big Bang.

\begin{figure}
\centering
\begin{subfigure}{.495\textwidth}
\centering
\includegraphics[scale=1]{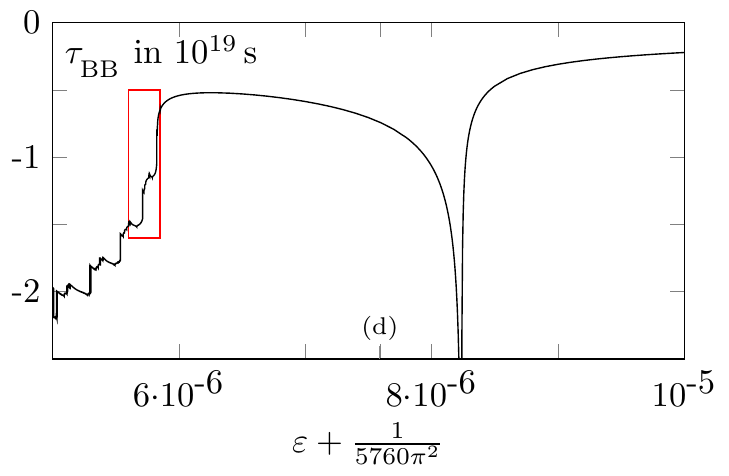}
\end{subfigure}
\begin{subfigure}{.495\textwidth}
\centering
\includegraphics[scale=1]{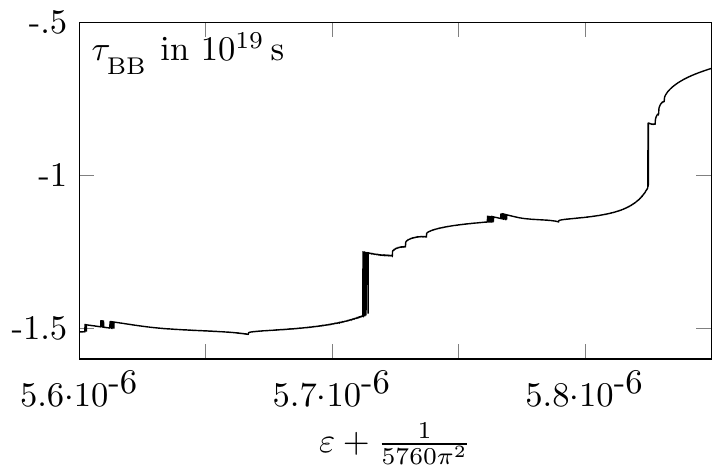}
\end{subfigure}
\begin{minipage}{.9\textwidth}
\captionsetup{format=plain, labelfont=bf}
\caption{$\tau_\textup{BB}$ as a function of $\varepsilon$ or, more precisely, its deviation from the critical value (i.e.\ for which $\exp\big(\frac{k_1}{k_2}\big)=1$), in the case $\xi=\frac{1}{6}$. The remaining parameters are set to the standard values from Section \ref{sec:General_discussion_on_the_decoupled_cSCE}. The red box in left plot marks the plot area of the right graphic. The tick (d) marks the $\varepsilon$-value of curve (d) in Figure \ref{conformally-coupled-case}.\label{taubigbangpole}}
\end{minipage}
\end{figure}

Figure \ref{taubigbangpole} shows $\tau_\textup{BB}$ as a function of $\varepsilon=3c_3+c_4$ for fixed $\xi=\frac{1}{6}$ and $H_0,q_0,\lambda_0$ and $\kappa$ as in Section \ref{sec:General_discussion_on_the_decoupled_cSCE}. 
We identify a divergence of $\tau_\textup{BB}$ as expected. 
At the left end of the plot, edited as a zoom in the right panel, we observe several discontinuities. Comparing the respective numerical solutions, we find that in each discontinuous step of $\tau_\textup{BB}$ the solution gathers another inflection point. The thick-lined part marks a discontinuity where $\tau_\textup{BB}$ jumps between two regions of continuity. 

In the left graphic of Figure \ref{taubigbang_2ndgraphic} we see the analog of Figure \ref{taubigbangpole} for some values $\xi\neq\frac{1}{6}$. Our observations match the expectations, namely that also for $\xi\neq\frac{1}{6}$ (but still close to $\frac{1}{6}$) the solutions show an oscillatory behavior which results in several inflection points with zero first derivative which, if shifted to the $a=0$-singularity by tuning $\varepsilon$, yields a divergence of $\tau_\textup{BB}$. The right panel displays this behavior as a function of both $\varepsilon-\varepsilon_\textup{crit}$ and $\xi$. 

Such negative poles of $\tau_\textup{BB}$  only exist up to a certain value of $|\xi-\frac{1}{6}|$, where the highest blue band (labelled $\Sigma$ in Figure \ref{taubigbang_2ndgraphic}) in the right panel of Figure \ref{taubigbang_2ndgraphic} meets the vertical axis on the left. Determining this value numerically, we find that for this $\xi$ value $|\xi-\frac{1}{6}|=\raisebox{3pt}{$\sqrt{\smash{\raisebox{-3pt}{$\scriptstyle\frac{1}{4320}$}}}$}$ (up to numerical error), that is, the maximum deviation of $\xi$ from $\frac{1}{6}$ such that our trace equation possesses exact de Sitter solutions specified in Section \ref{sec:General_discussion_on_the_decoupled_cSCE}. 
Recalling the discussion there, this is not surprising as the observed oscillations appear as decaying perturbations around the stable de Sitter solution.


To conclude, although the model introduced in this article \emph{can} solve the cosmic horizon problem, fine tuning of the renormalization constants is required and the resulting cosmologies are not close to the $\Lambda$CDM cosmological standard model. 


\section[Comparison with numerical $\Lambda$CDM model solutions]{Comparison with numerical $\mathbf{\Lambda}$CDM model solutions}
\label{sec:ComparisonLCDM}
In this section, we want to compare our model's solutions to the $\Lambda$CDM model's solutions with the parameters $\Omega_\textup{rad}=5.38\cdot10^{-5},~\Omega_\textup{dust}=0.315,~\Omega_\textup{DE}=0.685$ and $H_0=2.2\cdot10^{\textup{-}18}\,\frac{1}{\textup{s}}$ from \cite{pdg-data}. For this purpose, we fit our model's parameters to the $\Lambda$CDM solution using different measures of deviation. In this way, we obtain a rough idea of parameter regions of our model that produce `reasonable' cosmologies, despite the fact that a detailed investigation would require the inclusion of massive fields and therefore goes beyond the scope of this paper.  

\begin{figure}
\begin{subfigure}{.55\textwidth}
\centering
\includegraphics[scale=1]{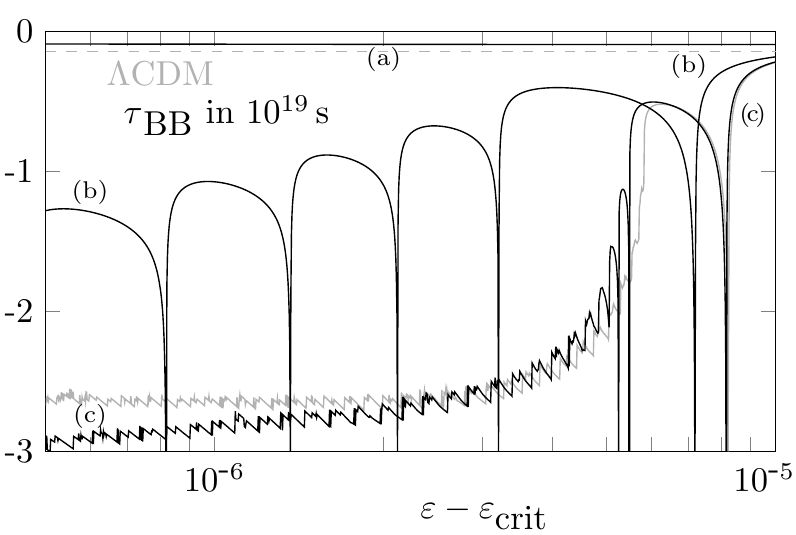}
\end{subfigure}
\begin{subfigure}{.44\textwidth}
\centering
\includegraphics[scale=1]{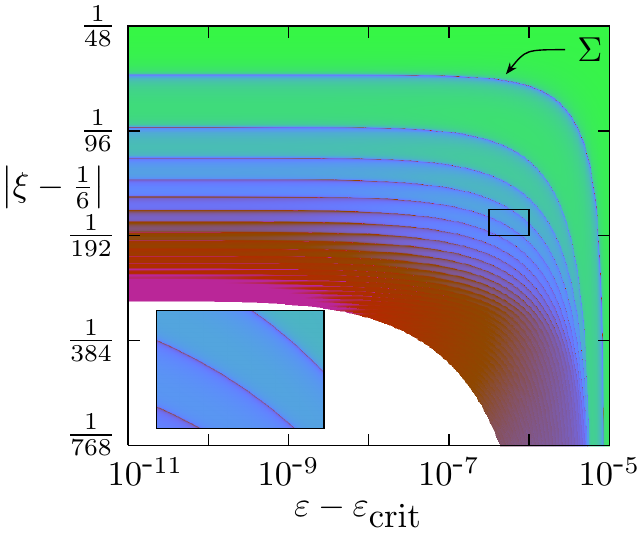}
\end{subfigure}

\begin{minipage}{.57\textwidth}
\captionsetup{format=plain, labelfont=bf}
\caption{The left graphic shows $\tau_\textup{BB}$ as a function of $\varepsilon-\varepsilon_\textup{crit}(\xi)$ for $\xi$-values close (but not equal) to $\frac{1}{6}$ which are listed to the right. The background curve in gray shows the reference curve from Figure \ref{taubigbangpole} with $\xi=\frac{1}{6}$. All remaining parameters are set to the standard values from Section \ref{sec:General_discussion_on_the_decoupled_cSCE}. The gray dashed line marks the analog value of the $\Lambda$CDM model as discussed in Section \ref{sec:General_discussion_on_the_decoupled_cSCE}. The right graphic shows $\tau_\textup{BB}$, now as a function of both $\varepsilon-\varepsilon_\textup{crit}(\xi)$ and $\xi$ together with a zoom of the boxed area. For a later purpose, we label the top right connected set of poles of $\tau_\textup{BB}$ by $\Sigma$.\label{taubigbang_2ndgraphic}}
\end{minipage}
\begin{minipage}{.39\textwidth}
\begin{tabular}{lll}
~&(a)&$\xi=\frac{1}{6}-\frac{1}{48}${\color{white}$\big|$}\\
~&(b)&$\xi=\frac{1}{6}-\frac{1}{192}${\color{white}$\big|$}\\
~&(c)&$\xi=\frac{1}{6}-\frac{1}{768}${\color{white}$\big|$}
\end{tabular}

\bigskip
~~\includegraphics[scale=1]{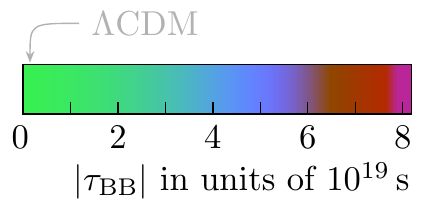}

\end{minipage}
\end{figure}

\subsection[$\Lambda$CDM uncertainty band]{$\mathbf{\Lambda}$CDM uncertainty band}\label{LCDM-uncertainty-band}
The $\Lambda$CDM parameters come with uncertainty errors, namely the 1-$\sigma$ uncertainties given by (cf.\ \cite{pdg-data})
\begin{align*}
\begin{aligned}
\Omega_\textup{rad}&=(5.38\pm0.15)\cdot10^{-5},\\
\Omega_\textup{DE}&=0.685\pm0.007,
\end{aligned}\hspace{1cm}
\begin{aligned}
	 \Omega_\textup{dust}&=0.315\pm0.007,\\
	 H_0&=(2.184\pm0.016)\cdot10^{\textup{-}18}\,\tfrac{1}{\textup{s}}.
\end{aligned}
\end{align*}
Bounded by these errors, we obtain a cuboid $Q$ in the $\Lambda$CDM parameter space. For each $y\in Q$, we denote the respective $\Lambda$CDM solution by $\alcdm(y):\R\to[0,\infty)$, wherefore we extend such a solution at the Big Bang and before by zero. By setting
\[
	a_\textup{max}(t):=\sup_{y\in Q}~\big(\alcdm(y)\big)(t)\hspace{1cm}\textup{and}\hspace{1cm}a_\textup{min}(t):=\inf_{y\in Q}~\big(\alcdm(y)\big)(t),
\]
we obtain an uncertainty band of the $\Lambda$CDM model in the $t$-$a$-plane that is compatible with 1-$\sigma$-errors in the single parameters.

As a numerical test of our model we want to determine a certain region of the $\xi$-$\varepsilon$-plane for fixed remaining parameters such that the solution $a$ fulfills
\begin{equation}
	a_\textup{min}(t)\le a(t)\le a_\textup{max}(t)\label{uncertainty-band-cond}
\end{equation}
for all $t\in\R$, where we likewise extend our solutions by zero before a Big Bang. 

\begin{figure}[t]
    \centering
    \begin{subfigure}{.495\textwidth}
    \centering
    \includegraphics{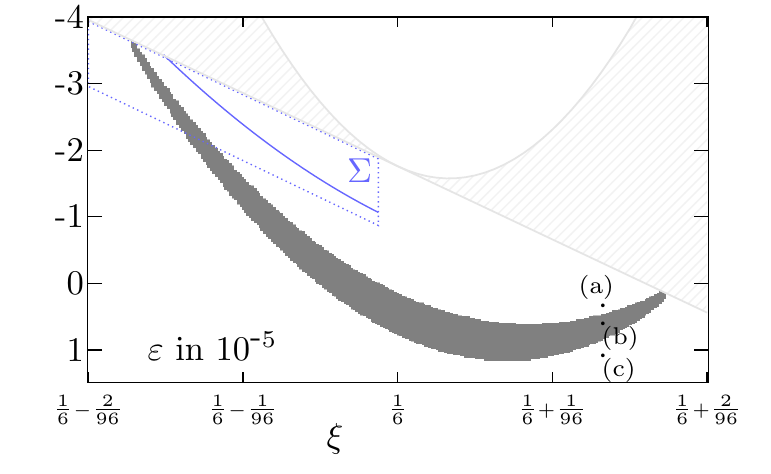}
    \end{subfigure}
    \hfill
    \begin{subfigure}{.495\textwidth}
    \centering
    \includegraphics{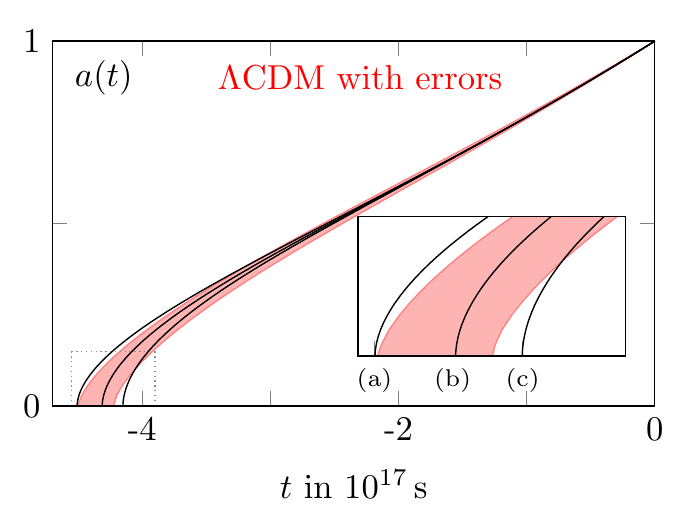}
    \end{subfigure}
    \begin{minipage}{.9\textwidth}
    \captionsetup{format=plain, labelfont=bf}
    \caption{The gray area in the left graphic marks the $\xi$-$\varepsilon$-points in which the solution of the cSCE fulfills \eqref{uncertainty-band-cond} for fixed remaining parameters as in Section \ref{sec:General_discussion_on_the_decoupled_cSCE}. The shaded area marks $\xi$-$\varepsilon$-points for which $a_\textup{crit}\in\big[\frac{1}{3001},1\big]$ holds. The points (a), (b) and (c) each mark an exemplary solution above, in and below the gray area, respectively. For orientation, the blue dotted parallelogram marks the boundaries of the right plot in Figure \ref{taubigbang_2ndgraphic} and therein the blue line marks the set $\Sigma$. The right graphic illustrates our numerical test by showing the $\Lambda$CDM uncertainty band in the $t$-$a$-plane in red together with the solutions of the cSCE corresponding to the parameter points (a), (b) and (c).\label{parspaceregion-figure}}
    \end{minipage}
\end{figure}
For fixed $H_0$, $q_0$, $\lambda_0$ and $\kappa$ as in Section \ref{sec:General_discussion_on_the_decoupled_cSCE}, Figure \ref{parspaceregion-figure} shows the region in the $\xi$-$\varepsilon$-plane where the solutions of our model satisfy \eqref{uncertainty-band-cond}. The shaded area in the left graphic of Figure \ref{parspaceregion-figure} marks the parameter region where $\exp\big(\frac{k_1}{k_2}\big)\in\big[\frac{1}{3001},1\big]$ holds. The left bound of said interval corresponds to the upper 
parabola-shaped bound of the shaded region. The right bound, in turn, corresponds to the linear lower bound of the shaded area, that is, it corresponds to $\varepsilon_\textup{crit}(\xi)$. 

As we can see, there exist parameters for which \eqref{uncertainty-band-cond} holds. They form a hook-shaped subset, narrowly distributed around the conformally coupled case $\xi=\frac{1}{6}$ and around the respective $\varepsilon_\textup{crit}(\xi)$. As mentioned before, our trace equation is symmetric under reflection at $\xi=\frac{1}{6}$ $-$ if we additionally adjust $\varepsilon$. Hence, the gray area has a symmetric shape if we skew the graphic in a way such that the values of $\varepsilon_\textup{crit}$ form a horizontal line. In the right graphic of Figure \ref{parspaceregion-figure} visualizes the uncertainty band defined by \eqref{uncertainty-band-cond} together with some sample curves. These represent the three possibilities of the solutions fulfilling the first inequality of \eqref{uncertainty-band-cond}, the second one or both of them, depending on whether the corresponding parameter point is below, above or inside the hook-shaped area, respectively. 

\subsection{Best parameter fit}
\label{Best-parameter-fit}
We next tune our model parameters in a way such that the solution is as close as possible to the $\Lambda$CDM solution. 

The major difficulty which prevents us from defining distance simply by some $L^p$-norm ($p\ge1$) is that the solutions of our models exist on variable  intervals. 
We therefore define a distance function  as follows. We first note that the $\Lambda$CDM solution is strictly monotonically increasing and continuous, hence invertible. Furthermore, the solutions of our model are invertible by the same argument, at least if we stay in the parameter regions of our `generic solution shape' of Section \ref{sec:General_discussion_on_the_decoupled_cSCE}. Thereby, we define the distance between the $\Lambda$CDM solution 
$\alcdm:(t_{\textup{BB},\Lambda\textup{CDM}}\,,\,\infty)\to\R$
(with the parameters from Section \ref{sec:General_discussion_on_the_decoupled_cSCE}) and a solution of our model
$
	a=a(\xi,\varepsilon,\kappa,H_0,q_0):(t_{\textup{BB},(\xi,\varepsilon,\kappa,H_0,q_0)},\infty)\to\R
$
by
\[
	d_{M,p}\big(a,\alcdm\big):=\left(~\int\limits_0^M \big|a^{\textup{-}1}(\alpha)-\alcdm^{\textup{-}1}(\alpha)\big|^p~\textup{d}\alpha\right)^{\nicefrac{1}{p}}
\]
with some $p\ge 1$ and some $M>0$ such that $a$ exists up to the value $M$. We solve the minimization problem
\begin{equation}
	\min_{(\xi,\varepsilon,\kappa,H_0,q_0)} d_{M,p}\big(a(\xi,\varepsilon,\kappa,H_0,q_0),\alcdm\big)\label{minimization-eq}
\end{equation}
where $(\xi,\varepsilon,\kappa,H_0,q_0)\in\R\times\R\times\R_{>0}\times\R_{>0}\times\R$. 

Note that a value of $M=1$ seems reasonable since in this context we consider the $\Lambda$CDM as a representation of experimental data and they are obviously measured at times where $a\le1$. Furthermore, we choose $p=2$ to suppress large deviations. 

By our previous discussions, we do not expect a unique minimum due to the symmetries of our trace equation under $\xi\mapsto\frac{1}{3}-\xi$ (and adjusting $\varepsilon$ to obtain the same deviation from $\varepsilon_\textup{crit}(\xi)$). 
Therewith, the minimum of course depends on the starting values for a downhill simplex algorithm. If the initial value of $\varepsilon$ is greater than $\varepsilon_\textup{crit}(\xi)$ (w.r.t the initial $\xi$) we would not expect the algorithm to be able to pass the $\varepsilon_\textup{crit}=0$-hypersurface\footnote{Here we refer to the hypersurface in the space $\R_\xi\times\R_\varepsilon\times(\R_{>0})_{\kappa}\times(\R_{>0})_{H_0}\times\R_{q_0}$ parameterized by $(\xi,\varepsilon_\textup{crit}(\xi))\in\R\times\R$ with $\xi\in\R$ in the first coordinates and arbitrarily in the remaining coordinates.} due to the behavior around the values $\varepsilon_\textup{crit}(\xi)$ presented in the previous sections. Also we would not expect the algorithm to pass the $\kappa=0$ hypersurface for the same reason. Furthermore, do we expect the deceleration parameter to remain in the interval $[-1,1]$ of reasonable values, since otherwise the inflection points presented in Figure \ref{deceleration-par-fig} fade in and yield a large $d_{M,p}$-distance for any choice of $(M,p)$.

\begin{figure}
\centering
\begin{subfigure}{.62\textwidth}
\centering
\includegraphics[scale=1]{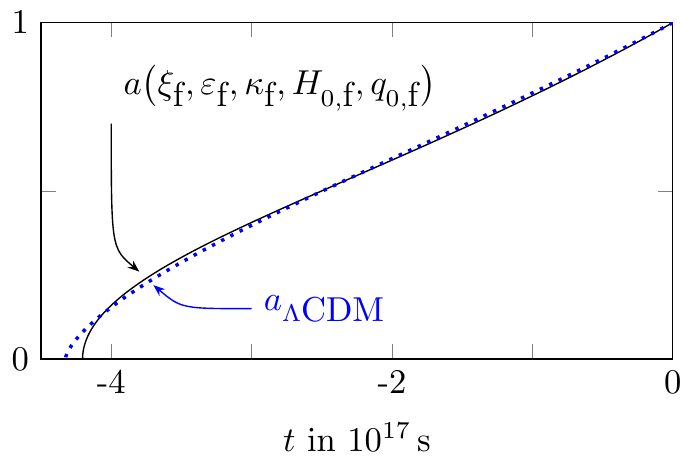}
\end{subfigure}
\hfill
\begin{subfigure}{.37\textwidth}
\begin{tabular}{lll}
$\xi_\textup{i}$&=&$\tfrac{1}{6}$\\
$\varepsilon_\textup{i}$&=&$1$\\
$\kappa_\textup{i}$&=&$2\cdot10^{42}$\\
$H_{0,\textup{i}}$&=&$2.1975\cdot 10^{\textup{-}18}\,\frac{1}{\textup{s}}$\\
$q_{0,\textup{i}}$&=&$-0.538$\\
~\\
$\xi_\textup{f}$&=&$0.1651$\\
$\varepsilon_\textup{f}$&=&$1.1152$\\
$\kappa_\textup{f}$&=&$2.6842\cdot10^{42}$\\
$H_{0,\textup{f}}$&=&$2.3507\cdot 10^{\textup{-}18}\,\frac{1}{\textup{s}}$\\
$q_{0,\textup{f}}$&=&$-0.7284$
\end{tabular}
\end{subfigure}
\begin{minipage}{.9\textwidth}
\captionsetup{format=plain, labelfont=bf}
\caption{One (local) minimizer of \eqref{minimization-eq} determined via a downhill simplex (Nelder-Mead) algorithm using the initial values $Z_\textup{i}$ on the right. The algorithm returned the values $Z_\textup{f}$ for which we see the solution on the left. As a reference the dotted line shows the $\Lambda$CDM solution.\label{fit-plot}}
\end{minipage}
\end{figure}

As a minimizer we find the parameters $Z_\textup{f}$ by using the exemplary initial values $Z_\textup{i}$ according to the table in Figure \ref{fit-plot}. The plot in Figure \ref{fit-plot} shows the respective solution to our model together with the $\Lambda$CDM solution (blue dotted). As we expected, we end up with a value of $\xi_\textup{f}$ close to $\frac{1}{6}$, with a value of $\varepsilon_\textup{f}\ge\varepsilon_\textup{crit}(\xi_\textup{f})$ and a value $q_{0,\textup{f}}\in[-1,1]$. Also $H_{0,\textup{f}}$ and $\kappa_\textup{f}$ remain close to $H_{0,\textup{i}}$ and $\kappa_\textup{i}$, respectively. Note that we only considered solutions with $\varepsilon>\varepsilon_\textup{crit}$ such that $a_\textup{crit}\in(0,1)$ is avoided and $a(t)$ covers $a$-values in the entire interval $(0,1)$. 


\subsection{The $\mathbf{\Delta N_\textup{eff}}$-test}
\label{Neff-test}

As one further method of comparing properties of our model to the respective properties of the $\Lambda$CDM model, we apply the $\Delta N_\textup{eff}$-test suggested by \cite{Hack:2015zwa} as a procedure to obtain limits for parameters in the SCE, see \cite{mangano2002precision,mangano2005relic}. Following the literature, we reparameterize the FLRW space-time with scale factor $a(t)$ via the red shift factor
\begin{equation}
	z(t)=\tfrac{1}{a(t)}-1\label{redshiftdef}.
\end{equation}
$N_\textup{eff}$ is the effective number of neutrino families, which can be related to $\Omega_\textup{rad}$ in \eqref{lcdm_equation} via
\begin{equation}
\label{eq:Neff}
\Omega_\textup{rad}=\Omega_\gamma\left(1+\frac{7}{8}\left(\frac{4}{11}\right)^{\nicefrac{4}{3}}N_\textup{eff}\right),
\end{equation}
where $\Omega_\gamma$ is given by the energy content of photons in the present universe at $t=0$. Note that from the observation of the cosmological microwave background (CMB) there are experimental values for $N_\textup{eff}$ stemming from the temperature spectrum of cosmic neutrinos in the CMB which slightly deviate from the thermal distribution. This deviation, in turn, can be computed from the energy distribution provided from a solution to the Boltzmann equation, in which the rate of expansion at the time of decoupling (between $z=3000$ and $z=1100$) enters \cite{mangano2002precision,mangano2005relic}. These calculations also derive the deviation from the number of neutrino families $N=3$ and also the prefactor $\frac{7}{8}\left(\frac{4}{11}\right)^{\nicefrac{4}{3}}=0.2271$.  The theoretical considerations, moreover, involve data from the Big Bang-nucleosynthesis (BBN) at $z\approx 10^9$, where the observed fraction of helium depends on the expansion rate.  The experimental findings are well compatible with the theoretical prediction $N_\textup{eff}=3.046$, i.e. $N_\textup{eff}=3.36\substack{+0.68 \\ -0.64}$ from the CMB power spectrum and $N_\textup{eff}=3.52\substack{+0.48 \\ -0.45}$ at BBN with 95\% confidence each, see \cite{Hack:2015zwa,Planck-Collab}.  The theoretical value for $N_\textup{eff}$ along with $\Omega_\textup{rad}=5.38\cdot 10^{-5}$ by \eqref{eq:Neff} results in $\Omega_\gamma=3.18\cdot 10^{-5}$.   

To connect $N_\textup{eff}$ to the rate of expansion, we define the difference of the squared normalized expansion rate to the theoretical prediction at the standard value for $N_\textup{eff}$ as
\begin{equation}
\label{eq:dHsq}
    \delta \left(\frac{H}{H_0}\right)^2(z):=\frac{1}{a(z)^2}\,\frac{\dot{a}(z)^2}{\dot{a}(0)^2}-\Omega_\textup{DE}-\Omega_\textup{dust}(1+z)^3-\Omega_\textup{rad}(1+z)^4,
\end{equation}
where in $\dot{a}(z)$ we first take the derivative with respect to $t$ and then reparameterize by \eqref{redshiftdef}.

Following \cite{Hack:2015zwa}, we define $\Delta N_\textup{eff}=N_\textup{eff}-3.046$ as the deviation of $N_\textup{eff}$ from the theoretical value given in \cite{mangano2002precision,mangano2005relic}.  Now, we can express the difference between the squared and normalized expansion rate $\big(\frac{H}{H_0}\big)^2(z)$ at the red shift parameter $z$ via    $\Delta N_\textup{eff}$ and obtain
\begin{equation}
	\Delta N_\textup{eff}(z)=\frac{1}{\Omega_\gamma}\,\frac{\delta \big(\frac{H}{H_0}\big)^2(z)}{0.2271(1+z)^4}.
\end{equation}
This parameterization of the observed difference in expansion can be used to check whether the error bounds, roughly $|\Delta N_\textup{eff}|\lesssim 1$ are fulfilled. As the BBN red shift for $z=10^9$ is hard to achieve numerically, we restrict to the CMB case and determine  $\Delta N_\textup{eff}(z=3000)$ as a function of the model parameters via simulations. $z=10^9$ for the BBN is beyond the capabilities of our solver.
For this purpose, we plot $\Delta N_\textup{eff}$ as a function of our parameters. 
Hereby, we restrict to the $\xi$-$\varepsilon$-plane and fix the remaining parameters as in Section \ref{sec:General_discussion_on_the_decoupled_cSCE}.

\begin{figure}
    \centering
    \includegraphics{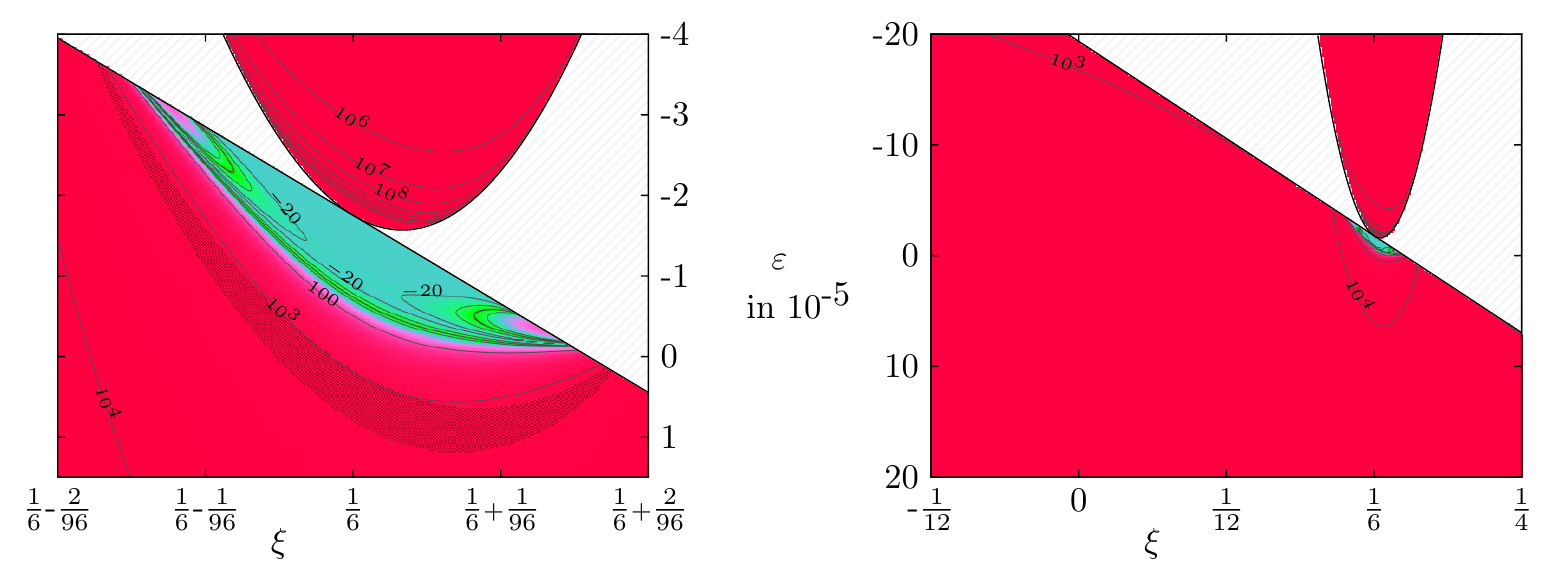}
    
    \begin{minipage}{.9\textwidth}
    \captionsetup{format=plain, labelfont=bf}
    \caption{$\Delta N_\textup{eff}(z=3000)$ as a function of $\xi$ and $\varepsilon$ for fixed remaining parameters as in Section \ref{sec:General_discussion_on_the_decoupled_cSCE}. The gray shaded area marks the parameter region where $\exp\big(\frac{k_1}{k_2}\big)\in\big[\frac{1}{3001},1\big]$ and hence $\Delta N_\textup{eff}(3000)$ does not exist. The colors are determined from the absolute value of $\Delta N_\textup{eff}(3000)$ and we emphasize the sign change along the green stripe, particularly there exists a curve whereon $\Delta N_\textup{eff}(3000)=0$. For orientation, the dashed area in the left graphic marks the parameter region represented in Figure \textup{\ref{parspaceregion-figure}}.\label{neff-figure}}
\end{minipage}
\end{figure}

Figure \ref{neff-figure} shows $\Delta N_\textup{eff}(z=3000)$ as a function of $\xi$ and $\varepsilon$ for points $(\xi,\varepsilon)$ in which $a_\textup{crit}\notin\big[\frac{1}{3001},1\big]$. Again, the diagonal straight line marks $\varepsilon_\textup{crit}(\xi)$ and if we skewed that line to be horizontal we would end up with a graphic that is symmetric with respect to reflection at $\xi=\frac{1}{6}$. On the other hand, the parabola shaped upper bound of the shaded area corresponds to $a_\textup{crit}=\frac{1}{3001}$. In the shaded area,
our solution does not reach $z=3000$ and the $\Delta N_\textup{eff}$-test does not make sense. 
In the left graphic of the figure we included the parameter region from Figure \ref{parspaceregion-figure}.
We find a small region where $\Delta N_\textup{eff}$ is smaller than the experimental error of 0.5 around $\xi=\frac{1}{6}$ and for rather small deviations of $\varepsilon$ from $\varepsilon_\textup{crit}$. 
An interesting feature of this graphic is that the level sets of $\Delta N_\textup{eff}(3000)$ is not too far away from the grey shaded region passing the test of Section \ref{LCDM-uncertainty-band}. Despite the fact that this region fails to pass the $\Delta N_\textup{eff}$ test by three orders of magnitude, one should keep in mind that our reduced model can only give a qualitative and preliminary insight into semiclassical cosmology.

As a main takeaway from this section, the $\Delta N_\textup{eff}$ seems to favour the region of small $\varepsilon$ and $\xi$ close to the case of conformal coupling $\xi=\frac{1}{6}$.

\begin{remark}
\normalfont
Another remarkable alignment is found between the present numerical test and the regions of a divergent $\tau_\textup{BB}$ from Section \ref{sec:Horizon}. The poles labeled by the set $\Sigma$ apparently match with the $\Delta N_\textup{eff}(3000)=0$-level set quite well.
\end{remark}


\section{Conclusion and Outlook}
\label{sec:Conclusion}

In this work we investigated cosmological solutions of the SCE for massless quantum fields in special Minkowski-like states. In such states, the dynamical degrees of freedom from the scale factor decouple from the dynamics of the quantum state, as such states come with a vanishing `tower of moments' in the sense of \cite{siemssen-gottschalk}. While this phenomenon was well known in the conformally coupled case \cite{Starobinski}, we observe some new cases here including also non conformally coupled fields. We thus retrieve new cosmological models from the solutions of the massless SCE.

We provided a detailed numerical study of these new cosmological models.  Typical solutions show a radiation like Big Bang in the early universe {\ifanonymousforreviewreply(\color{blue}\fi{}sufficiently far remote from the Planck scale)} in conjunction with a Dark Energy-like behavior for the late universe. In our models, the late time universe Dark Energy phase is observed without introducing a cosmological constant, neither directly nor through a renormalization constant. Such models expose a smooth transition in the state equation connecting energy and pressure that ranges from the ratio $\frac{1}{3}$ (corresponding to radiation) to $-1$ (corresponding to Dark Energy). 

We also investigate special parameter settings that give rise to a solution of the cosmological horizon problem as proposed by \cite{pinamonti2011initial}. While we give numerical evidence that such solutions exist, we also see that this behavior requires parameter tuning and is not stable under small parameter variations.

A large part of this work is concerned with a numerical comparison of our cosmological models with the $\Lambda$CDM standard model of cosmology. Evidently, we observe deviations in the cold dark matter dominated phase of $\Lambda$CDM cosmology, which in turn is to be expected for a massless quantum field. However, we show that a parameter fit of our cosmological models to the $\Lambda$CDM cosmology endowed with physical parameters \cite{Planck-Collab} remains in the strip of observational 1-$\sigma$-uncertainty of the  $\Lambda$CDM model. Also, we identify `physical' parameter regions that comply with the $\Delta N_\textup{eff}$-test as suggested by \cite{Hack:2015zwa}. Despite that both parameter regions do not have an intersection, they are positioned close to each other in parameter space. Interestingly, these tests seem to favour coupling $\xi$ close to conformal coupling $\xi=\frac{1}{6}$ rather than minimal coupling $\xi=0$ and a small value of the renormalization constant $\varepsilon=3c_3+c_4$. 

While we have gathered evidence that semiclassical cosmology, even without cosmological constant, can produce interesting cosmologies that are not too far from the standard cosmology, further extension and refinement of the model seem to be in order. Obviously, massive fields should be incorporated and also fields with higher spin and Fermi statistics \cite{Hack:2010iw}.

\vspace{.3cm}
\ifanonymousforreviewreply
\else
\noindent\textbf{Acknowledgement.}\\ The authors thank T.-P.\ Hack, N.\ Pinamonti and P.\ Meda for interesting discussions. Moreover, the authors are grateful towards the referees for useful remarks.
\fi
\bibliography{bibliography}{}
\bibliographystyle{plain}

\end{document}

%% file: commands.tex
\newcommand{\munu}{\ensuremath {{\mu\nu}}}
\newcommand{\R}{\ensuremath\mathbb{R}}

\newcommand{\ddt}{\ensuremath\frac{\textup{d}}{\textup{d}t}}

\newcommand{\Tren}{\ensuremath\big\langle {T^{\textup{ren}}}\big\rangle_\omega}
\newcommand{\Tmunuren}{\ensuremath\big\langle {T_\munu^{\textup{ren}}}\big\rangle_\omega}
\newcommand{\Tmunurenen}{\ensuremath\big\langle {T_{00}^{\textup{ren}}}\big\rangle_\omega}

\newcommand{\mathcalm}{\ensuremath\raisebox{-1pt}{\includegraphics[scale=1.1]{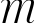}}}

\newcommand{\alcdm}{\ensuremath a_{\Lambda\textup{CDM}}}